\documentclass[11pt]{article}
\usepackage{epsfig}
\usepackage{amssymb}
\usepackage{latexsym} 
\usepackage{amsmath}
\usepackage{epic}
\usepackage[margin=1in]{geometry}
\usepackage{color}
\usepackage{graphicx}
\usepackage[ruled, vlined]{algorithm2e} 

\newtheorem{theorem}{Theorem}
\newtheorem{definition}{Definition}
\newtheorem{lemma}{Lemma}
\newtheorem{proposition}[theorem]{Proposition}
\newtheorem{corollary}{Corollary}
\newtheorem{claim}{Claim}

\newenvironment{proof}{\noindent\textbf{Proof: }\ignorespaces}{}
\newcommand{\qed}{\hspace*{\fill}$\Box$\medskip}

\newcommand{\eat}[1]{}

\def\con{\subseteq}
\def\Mt.{\mathy{T(M)}}
\def\mathy #1{\ifmmode {#1} \else{$#1$}\fi}
\def\P.{{$\cal P$}}
\def\N.{{$\cal N$}}
\def\C.{\mathy{\cal C}}
\def\A.{\mathy{\cal A}}
\def\B.{\mathy{\cal B}}

\def\remark #1{\noindent{\bf Remark:} #1\\}
\def\numremark #1 #2{\noindent{\bf Remark #1:} #2\\}

\long\def\claim #1{\bigskip\noindent{\bf Claim} {\it #1}\bigskip}
\def\case #1 #2{\bigskip\noindent{\em Case #1: #2}\par}
\def\set #1#2{\{ #1:#2 \}}
\newcommand{\bproof}{\noindent{\bf Proof: }}

\newcommand{\ecproof}{\hfill $\spadesuit$\\}
\long\def\claim #1{\bigskip\noindent{\bf Claim} {\it #1}\bigskip}
\def\i{($i$) } \def\xi{($i$)}
\def\ii{($ii$) } \def\xii{($ii$)}
\def\iii{($iii$) } \def\xiii{($iii$)}

\title{ \Large A Model for Minimizing
Active Processor Time}

\author{
Jessica Chang\thanks{
Dept. of Computer Science and Engineering, University of Washington, Seattle WA 98195,
\texttt{jschang@cs.washington.edu}.
Research supported by an NSF Graduate Research Fellowship and
NSF CCF-1016509.}
\and
Harold N. Gabow \thanks{University of Colorado, Boulder CO 80309, 
\texttt{hal@cs.colorado.edu}.}
\and
Samir Khuller\thanks{
Dept. of Computer Science, University of Maryland, College Park MD 20742, \texttt{samir@cs.umd.edu}.
Research supported by NSF CCF-0728839, NSF CCF-0937865 and a Google Research Award.} }
\date{}

\begin{document}
\thispagestyle{plain}
\maketitle

\begin{abstract} 
We introduce the following elementary 
scheduling problem. We are 
given a collection of $n$ jobs, where each job 
$J_i$ has an integer length $\ell_i$ as well as 
a set $T_i$ of time intervals in which it can be 
feasibly scheduled.  
Given a parameter $B$, the processor can schedule up to $B$ jobs
at a timeslot $t$ so long as it is ``active'' at $t$.
The goal is to schedule all the jobs in the fewest 
number of active timeslots.  The machine consumes 
a fixed amount of energy per active timeslot, 
{\em regardless} of the number of jobs scheduled 
in that slot (as long as the number of jobs is non-zero).
In other words, subject to $\ell_i$ units of each 
job $i$ being scheduled in its feasible region and 
at each slot at most $B$ jobs being scheduled, we 
are interested in minimizing the total time during 
which the machine is active.  We present a linear 
time algorithm for the case where jobs are unit length 
and each $T_i$ is a single interval.  For general $T_i$, 
we show that the problem is $NP$-complete even
for $B=3$. However when $B=2$, we show that it
can be efficiently solved.
In addition, we consider a version of the 
problem where jobs have arbitrary lengths and 
can be preempted at any point in time.
For general $B$, the problem can be solved by linear 
programming.  For $B=2$, the problem amounts to 
finding a triangle-free 2-matching on a special graph.
We extend the algorithm of Babenko et. al. \cite{COCOON10}
to handle our variant, 
and also to handle non-unit length jobs.
This yields an $O(\sqrt L m)$ time algorithm
to solve the preemptive scheduling problem for $B=2$,
where $L = \sum_i \ell_i$.
We also show that for $B=2$ and unit length jobs,
the optimal non-preemptive schedule has $\le 4/3$ 
times the active time of the optimal
preemptive schedule; this bound extends to 
several versions of the problem when jobs
have arbitrary length.
\end{abstract}

\thispagestyle{plain}
\section{Introduction}
Power management strategies have been widely studied in the
scheduling literature \cite{Albers09,Albers10,IP,YDS,IGS,Baptiste}.  
Many of the models are motivated by
the energy consumption of the processor.
Consider, alternatively, the energy consumed 
by the operation of large storage systems.
Data is stored in memory which may be turned 
on and off \cite{SoCC2010}, and each task or 
job needs to access a subset of data items to run. 
At each time step, the scheduler can work on
a group of at most $B$ jobs. The only requirement is
that the memory banks containing the 
required data from these jobs be turned on. 
The problem studied in this paper is the special 
case where all the data is in one memory bank.
For even special cases involving multiple memory 
banks, the problem becomes $NP$-complete: if each job 
needs access to multiple memory banks
in order to be satisfied, via a reduction from 
the $k$-densest subgraph problem, it is $NP$-complete 
to determine whether there exists a schedule
satisfying $C$ jobs and being active for at most 
$A$ units of time.

We propose a simple model for measuring energy 
usage on a parallel machine.  Rather than focusing 
on conventional metrics measuring the quality of 
the schedule, we focus on problems motivated by
energy savings in ``efficient'' schedules.

In many applications, a job has many 
intervals of availability because, e.g., it 
interfaces with an external event such as 
a satellite reading or a recurring broadcast.  
The real-time and period scheduling literatures
address problems in this space.
More broadly, tasks may be constrained by 
user availability, introducing irregularity 
in the feasible intervals.  Our model is defined
generally enough to capture jobs of this nature.

More formally, we are given a collection of 
$n$ jobs, each job $J_i$ having an integer 
length $\ell_i$ and a set $T_i$ of time intervals with integer boundaries
in which it can be feasibly scheduled.  
In particular, 
$T_i = \{ I_k^i=[r_k^i, d_k^i] \}_{k=1}^{m_i}$ is a non-empty
set of disjoint intervals.  Note that if $m_i = 1$, then
we can think of job $J_i$ as having a single release time
and a single deadline.
For ease of notation, we may sometimes refer to job $J_i$ as job $i$.
In addition, time is 
divided into unit length timeslots
and for a given parallelism parameter $B$, 
the system (or machine) can schedule up to $B$ jobs in a single timeslot.
If the machine schedules any jobs at timeslot $t$, we say that
it is ``active at $t$''.
The goal is to schedule all jobs, i.e. schedule them 
within their feasible regions, 
while minimizing the number of slots during which the machine is active. 
The machine consumes a fixed amount of energy per active slot. 
In other words, subject to each job $J_i$ being 
scheduled within its feasible region $T_i$, 
and subject to at most $B$ jobs being scheduled 
at any time, we would like to
minimize the total active time spent scheduling the jobs.  
Note that there may be instances when there is no feasible schedule 
for all the jobs.  However, this case is easy to detect.
Note that for a timeslot significantly large (e.g. on the
order of an hour), any
overhead cost for starting a memory bank is negligible compared
to the energy spent being ``on'' for that unit of time.

To illustrate this model in other domains,
consider the following operational problem.
Suppose that a ship can carry up to $B$ cargo
containers from one port to another. 
Jobs have delivery requirements leading to release times
and deadlines. Finding an optimal schedule corresponds to the
minimum number of times we need to send the ship to deliver all the
packages on time. The motivating assumption is that it costs roughly the
same to send the ship, regardless of load and that there is an
upper bound on the load.

We could also consider this as a basic form of ``batch'' processing
similar to the work initiated by Ikura and Gimple \cite{IG}. 
Their algorithm is designed to minimize completion time for 
batch processing on a single machine for the 
special case of {\em agreeable}\footnote{When the ordering of jobs by release
times is the same as the ordering of jobs by deadlines.}
release times and deadlines. 
Baptiste \cite{Baptiste00} extended the Ikura and Gimple results 
to general release times and deadlines and an efficient algorithm was 
recently given by Condotta et. al.\cite{CKS}. 
All of these works focus merely on trying to find 
a feasible schedule (which then can be 
used as a subroutine to minimize maximum lateness). However
in our problem, in addition
we wish to minimize the number of batches. 

In the scheduling literature, often problems with unit processing times 
are trivial since they can be solved using matching techniques. However,
in different models which allow for overlap in job satisfaction, 
e.g. broadcast scheduling \cite{EH,CK,CEGK}, the problems often
turn out to be $NP$-complete; in fact, several variants of broadcast
scheduling have been shown to be $NP$-complete \cite{CEGK}. 
The problem considered in this paper also contains
an element of ``overlap'' since we can schedule up to $B$
jobs in a slot at unit cost and wish to minimize the number
of active slots.  

For the cases of unit length jobs and those in 
which jobs can be preempted at integral time points,
our scheduling problem can be modeled as a 
bipartite matching problem in which each node 
on the left needs to be matched with a node on the
right. Each node on the right has a capacity of $B$, 
and we are interested in minimizing the 
number of nodes on the right that have non-zero degree.
This problem can easily be shown to be $NP$-hard.
Hence it is slightly surprising that for unit 
length jobs with each $T_i$ being a single
interval, we can develop 
a fast algorithm to obtain an optimal solution to 
the scheduling problem defined above\footnote{The problem 
can be solved in $O(n^2T^2(n+T)$ time using Dynamic 
Programming as was shown by Even et. al. \cite{Even},
albeit the complexity of their solution is high. Their 
algorithm solves the problem of stabbing a collection 
of horizontal intervals with the smallest number of 
vertical stabbers, each stabber having a bounded capacity.}.  
Our algorithm is an almost greedy scheme, which intuitively abides 
by a lazy activation principle: schedule jobs in batches of size up to 
$B$ delaying the batch as long as possible.  At each step, we select 
``filler'' jobs (with later deadlines) to fill slots which otherwise would
have at least one and less than $B$ jobs, based on an Earliest Deadline 
First (EDF) strategy.  The algorithm as described does not quite work, 
since we may schedule some jobs using the lazy activation principle 
and later discover that these jobs should have been scheduled earlier 
to make space for other jobs with later deadlines. One way to address 
this problem is to dynamically re-assign jobs to time slots.  Our first 
attempt was based on this idea, but it resulted in a 
slower algorithm with a more complicated analysis.  However, we are 
able to address this issue by pre-processing the jobs to create 
a new instance with ``adjusted" deadlines, so that at most 
$B$ jobs have the same deadline. Then, no re-assignment 
of jobs is required.
%
%
As we will see, for infeasible instances, this algorithm has the
additional property that it will schedule the maximum number of jobs.

In particular, as our paper shows, even the problem where $B=2$ has a lot 
of structure due to its connection with matchings in graphs. We anticipate that 
this structure will be useful in the design of improved approximation algorithms
for $B>2$ (an $O(\log n)$ approximation algorithm follows from the work
of Wolsey \cite{Wolsey}).

\noindent
{\bf Main Results:}

\begin{enumerate}
\item For the case where jobs have 
unit length and each $T_i$ is a single interval,
we first develop an algorithm whose running time is $O(n \log n)$.
We then show how to improve its running time to linear.
Our algorithm takes $n$ jobs as input with integral release 
times and deadlines and outputs a schedule with the smallest 
number of active slots.   
The algorithm has the additional property that for
infeasible instances, it schedules the maximum number of jobs.
We also note that the slotted aspect of the time model is but a technical
convenience.  It can be shown without loss
of generality that time is slotted when job lengths, release times
and deadlines are integral (Section~\ref{sec:main}).  

\item 
When the release times and deadlines are not integral,
non-preemptively scheduling unit length jobs to
minimize the number of batches can be solved optimally 
in polynomial time via dynamic programming (Section~\ref{sec:dp}). 
This objective differs from active time:
a batch must start all its jobs at the same time and
the system may work on at most one batch at a time.
Even so, scheduling unit length jobs with integral
release times and deadlines to minimize active time
is clearly a special case of this.  We extend the 
result to the case when we have a budget on the number 
of active slots (Section~\ref{sec:dp2}).

\item In addition, we consider the generalization 
to arbitrary $T_i$.  
This problem is closely related to vertex cover with hard capacities,
the $k$-center problem and capacitated facility location, 
all classic covering problems.  In
particular, for the special case where every job is feasible in
exactly two timeslots, there is a LP-rounding 2-approximation,
which is implied from the vertex cover result in \cite{GHKKS}.
The complexity 
of the problem depends on the value of $B$, since for any fixed $B \ge 3$,
the problem is $NP$-hard. When $B=2$ this problem can be 
solved optimally  in  $O(m \sqrt{n}  )$ time where $m$ is the 
total number of time slots which are feasible for some job (Section~\ref{sec:disjoint}). 
We show that this problem is essentially equivalent to the maximum 
matching problem computationally.
In addition, we show
that this algorithm can be extended to the case of non-unit length jobs when
a job can be scheduled in unit sized pieces (Section~\ref{sec:nonunit}).

\item
We also consider a version of the problem when 
jobs have arbitrary lengths and can be preempted at any point in time,
i.e. not just at integer time points.
For general $B$ the problem can be solved by linear programming.
For $B=2$ the problem amounts to finding a maximum 
triangle-free 2-matching on a special graph.
Babenko et. al. present an elegant algorithm showing that 
a maximum cardinality triangle-free 2-matching
can be found in the same time as a maximum cardinality
matching \cite{COCOON10}.  We extend it for our 
scheduling problem to show that when $B=2$ and jobs
have arbitrary integral length, an optimal preemptive
schedule can be found in $O(\sqrt L m)$ time, for
$L$ the total length of all jobs.  Any preemptions occur at
integral or half-integral times.  

\item
In Section 7, we also give a tight
bound on the gain from arbitrary preemption: an optimal
schedule allowing only preemption at integral times uses
at most $4/3$ the active time of the optimal preemptive schedule.
We also note that this bound is best possible 
since there is a trivial example with three unit 
jobs where the optimal schedule which allows preemption
only at integer points uses two slots, 
and if we allow 
arbitrary preemptions, these jobs
can be scheduled in 1.5 slots giving the ratio of $4/3$.

\end{enumerate}

\subsection{Related Work}

A classical problem related to our work is the well known
``Scheduling unit jobs on $B$ processors with precedence constraints'',
in which $n$ unit jobs are given with 
precedence constraints and the goal is
to schedule the jobs on $B$ 
processors to minimize the maximum completion
time.  Again this can be viewed as minimizing the active time.  
For arbitrary $B$ the problem is $NP$-complete
\cite{GJ}. For fixed $B$, the problem is known to be  W[2] hard \cite{BF}. 
For the case where $B=2$, this problem can 
be solved optimally in polynomial
time \cite{Fuji,Gabow}.  Garey and Johnson \cite{GJ1,GJ2} consider 
the problem of scheduling unit jobs 
with integer release times and deadlines 
with precedence constraints, providing 
polynomial time algorithms for $B=2$.  
In addition, their algorithm finds a 
schedule with minimum lateness.   
This was extended to the case of real 
release times and deadlines by Wu and 
Jaffar \cite{WuJaffar}, who gave an $O(n^4)$ algorithm.
Their primary technique involves computing 
successor-tree-consistent deadlines, 
which effectively upper bounds the 
latest completion time for each job.  
Successor-tree-consistency allows 
the optimal schedule to be computed 
via a slight variation to Simon's forward 
scheduling algorithm for independent 
unit length jobs.  For scheduling unit length jobs with arbitrary
release times and deadlines on $B$ processors 
to minimize the sum of completion times, Simons and Warmuth
\cite{SW} extended the work by Simons \cite{S} giving an
algorithm with running time $O(n^2 B)$ to find a feasible
solution.  For constant $B$, the running time is improved
in \cite{LOQ11}.

A closely related problem of minimizing busy time 
has been recently studied by Khandekar, Schieber,
Shachnai and Tamir \cite{KSST}. In the busy time problem, jobs of arbitrary
length have individual demands $r_j$. 
The jobs have release times and deadlines 
and need to be scheduled in batches,
with the additional requirement that 
the total demand of jobs in the batch at any point of time 
is at most a given value. Each batch is scheduled on a single machine.
The busy time of a batch is defined
as the busy time of the machine that schedules it, i.e. the duration of the
earliest start time to latest end time of jobs in the batch.  We highlight that
their model permits access to an unbounded number of machines; thus every
instance is feasible, albeit with potentially high cost. 
Our problem is slightly different in that we have uniform demands  (as in \cite{IPDPS}), 
and we do not
have an unbounded number of machines. The non-unit length generalization makes the
problem $NP$-hard, even for the uniform demand case \cite{IPDPS,WZ}. 
In \cite{IPDPS}, the authors consider the uniform 
demand case and present a 4-approximation as well as results
for the special case where the jobs are {\em interval} jobs, i.e. the processing time
is exactly the length of the interval. In \cite{KSST}, they consider a more general
problem and develop an approximation algorithm with a factor of 5. Their main
idea is to first ``cluster'' the jobs with the assumption that each
batch has infinite capacity and then fix this as the position
of the job by modifying the release time and deadline, thus converting it
to an {\em interval} job. The main algorithm then partitions the
jobs by demand into two categories and uses a greedy method to schedule the
jobs.  A number of applications are mentioned in \cite{KSST,IPDPS}.

Baptiste \cite{Baptiste} examines a related problem of ``min gap'' scheduling unit length jobs 
on a single processor to minimize the number of idle intervals; in this model, 
the algorithm determines when the processor sleeps. They give an 
optimal dynamic programming algorithm which builds from a dominance
property of the optimal offline schedule.  
Baptiste, Chrobak and D\"urr in \cite{BCD} improve the running time; 
their algorithm in fact applies to the generalized problem in which jobs 
have arbitrary processing times. This work was subsequently extended to handle
multiple processors by  Demaine et. al. \cite{DGHSZ}, who also
provide an approximation algorithm for the case where each job
has multiple intervals in which it can be scheduled.  They also give
$\log n$ lower bounds on the approximation ratio.  
The cost function in this lower bound does not apply to 
the problems studied in this paper.

Also related is the \textit{dynamic speed scaling} problem, 
in which the scheduler determines the 
non-negative speed at which the processor runs.
For a single processor,
Yao et al. \cite{YDS} give an exact offline solution that 
minimizes the total power consumption when the power
is a convex function of the speed.
Irani et al. \cite{IGS} study an extended problem in which
the machine can also be put into the ``sleep'' state, during
which period no cost is incurred other than the constant wake-up
cost.
They present a 2-approximation and an 
$O(1)$-competitive algorithm in the online setting.  
The problem was recently shown to be $NP$-complete
and the approximation improved to 4/3 by Albers
and Antoniadis \cite{AA12}.
Despite the significant 
results in \cite{IGS}, its authors acknowledge 
that a continuous power function is unrealistic; 
in practice, systems run at 
a finite number of potential speeds.  Our work is the 
special case in which power is represented by a step 
function.

Li and Yao \cite{LY} consider a discretized version of the problem, in which the 
system may operate at one of a finite number of speeds.  Their algorithm is exact 
and runs in time $O(dn \log n)$, where $d$ is the number of possible speeds.  The main
idea is to first partition the jobs, and then to determine the speeds of these jobs, 
partition by partition. However, their model is not quite the same as ours; despite 
the discretization of the speeds, they still assume that the underlying power function is convex 
and therefore cannot capture the step from speed 0 to positive speed.

Demaine et. al. \cite{DGHSZ} investigated problems involving multiple processors
and more generally, multiple feasible intervals for jobs.  
For the multiple processors settings, they provide a polynomial-time algorithm
which minimizes the total number of gaps in the schedule.  The
algorithm also minimizes the total transition energy plus total time in active slots, over
multiple processors.  Notice that this setting is not quite a generalization of ours, since
the total active time is summed over each processor.
Unlike our cost model, it is cheaper to activate fewer rather than more processors at any given time. 
Finally, they give for the multi-interval setting a 
$(1 + (\frac{2}{3}+\epsilon)\alpha)$-approximation, where $\alpha$ is the cost
to transition to the active state.  Notice that their cost model is very closely related to
ours, with the exception that we consider settings in which the startup cost is negligible enough
to be assumed zero.  Thus, their lower bounds on the approximation ratio, which
explicitly assume a non-zero $\alpha$, do not apply to our problem.

Demaine and Zadimoghaddam \cite{DZ} recently studied the problem of minimizing 
energy consumption in schedules over multiple processors.  Their model is quite 
general in that the feasible time slots in which each unit length job may be scheduled 
may not comprise a single time interval.  Also, each processor has an arbitrary (unrelated) 
energy function and can (but doesn't have to) go to a sleep state.  They provide a 
$O(\log n)$-approximation for this problem by first proving a general result for the 
submodular maximization problem with budget constraints, and then reducing their 
scheduling problem to a matching problem on a bipartite graph with a submodular 
matching function.  They also show that the problem is Set-Cover-hard, thus 
demonstrating the tightness of their result.  

We refer the reader to surveys \cite{IP,Albers09,Albers10} for a 
more comprehensive overview of the latest scheduling results for power management problems.

\section{Lazy Activation for Unit Jobs and Single Execution Windows}
\label{sec:main}

We first provide a high level description of the algorithm, followed
by pseudo-code and the proof of optimality. We may assume that
the instance is feasible, since this is easy to check
by an EDF computation.

Denote the distinct deadlines
by $d_{i_1} < d_{i_2} < \ldots < d_{i_k}$, and let
$S_p$ be the set of jobs with deadline $d_{i_p}$.  Then 
$S_1 \cup S_2 \ldots \cup S_k$ is the entire set of jobs where the
deadlines range from time $1$ to $T=d_{i_k}$.
It is not hard to argue that w.l.o.g. $T$ is $O(n)$. 

We process the jobs in two phases.  In Phase I we scan the 
jobs in order of {\em decreasing} deadline. We do not
schedule any jobs, but only modify the deadlines  of jobs to create a
new instance, whose optimal solution is equivalent to
that of the original instance. The desired property of the new instance is that at most $B$
jobs will have the same deadline.
Process the time slots from right to left. At slot $D$, let $S$ be the 
set of jobs that currently have deadline $D$. From $S$, select
$\max(0,|S| - B)$ jobs with earliest release times and decrement
their deadlines by one.
If $|S|  \le B$ then we do not modify the deadlines of jobs in $S$.
(Note that a job may have its deadline reduced multiple times since
it may be processed repeatedly.)
 
Assume for simplicity's sake that after the first phase,  $S_p$ refers to the jobs
of (modified) deadline $d_{i_p}$. 
We now describe Phase II in which jobs are actually scheduled. 
Initially all jobs are {\em unscheduled}. As the algorithm
assigns jobs to active time slots, we change the status of jobs to {\em scheduled}. 
Once a job is scheduled, it remains scheduled. Once a slot is declared active, 
it will remain active for the entire duration of the algorithm.

Our algorithm, in general, schedules jobs in increasing order by deadline.
As we will see shortly, this is not quite true since 
we may schedule some jobs with later deadlines whenever there is available 
space. 
In precise terms, we schedule the time slots
$d_{i_p}$ from left to right. To schedule $d_{i_p}$, if there
are still unscheduled jobs with that deadline
we schedule them. If fewer than $B$ jobs get scheduled,
we schedule additional jobs that are available.
Job $j$ is available if it is currently unscheduled
and has $d_{i_p} \in [r_j,d_j)$. We schedule these
available jobs EDF, until the slot is full or
no more jobs are available.

\subsection{Formal Algorithm Description}
Let $B > 0$ be the number of jobs that the system can satisfy in a single time slot.
For every job $j$, denote $j$'s release time and deadline by $r_j$ and $d_j$, respectively.
We index the $n$ jobs in order of increasing deadline, i.e. such that
$d_1 \leq d_2 \leq \ldots \leq d_n = T$. 
Normalize to 0 the earliest release time of any job.
Then without loss, all feasible schedules are active only within the interval $[0,T]$.

In the first phase we scan the jobs from right to left in decreasing deadline order.
At each step we consider the set of jobs with a common deadline and leave
up to $B$ jobs with the latest release times untouched.  We modify the
deadlines of the rest of the jobs in this set, decrementing them each by one, 
and then continue processing the jobs.
 
The algorithm for the second phase
simultaneously maintains a set $W$ of active time slots and a set 
$J$ of satisfied jobs, both of which are initially empty.  In each iteration, 
we look at the unsatisfied job $j^*$ of earliest deadline, i.e. $j^* = \arg \min_{j \notin J} d_j$.  
Let $d^*$ be $j^*$'s deadline, and let $J^*$ be the set of all 
unsatisfied jobs with the same deadline $d^*$. 
The algorithm activates the latest possible time slot which can satisfy $J^*$ 
and adds it to $W$.  Only one slot is needed to satisfy $J^*$ since $|J^*| \leq B$.

The algorithm assigns jobs to the newly activated time slot $t$
first adding $J^*$ to $J$ and assigning those jobs to slot $t$.
If fewer than $B$ jobs are available, we  fill the remaining space by selecting 
available (filler) jobs from the remaining set of unscheduled jobs, again selecting 
based on EDF. These filler jobs will then be added to $J$.

The following pseudocode formalizes the above description of the second phase.
{\em SelectFillers} takes as input the set of active
slots $W$ and the set of scheduled jobs $J$, returning the set
of filler jobs $J'$. These jobs are then
added to the set of scheduled jobs.

\begin{algorithm}
$J \leftarrow \emptyset$; $W \leftarrow \emptyset$ \;
\While{$\exists \; j \notin J$}{
	$d^* \leftarrow \arg \min_{j \notin J} d_j$ \;
	$J^* \leftarrow \{ j \notin J : d_j = d^* \}$ \;
	$W \leftarrow W \cup d^*$ \;  
        $J \leftarrow J \cup J^*$ \;
	$J' \leftarrow SelectFillers (W, J)$ \; 
	$J \leftarrow J \cup J'$ \;
}

\caption{Lazy Activation Algorithm}
\label{ALG}
\end{algorithm}

\begin{algorithm}
Choose available fillers based on EDF to fill the $B-|J^*|$ empty spots.

\caption{SelectFillers($W,J$)}
\label{ALG2}
\end{algorithm}

\subsection{Analysis of the Algorithm}

It  is easy to implement our algorithm in running time $O(n \log n)$ using
standard data structures. What is not completely obvious is why it computes
an optimal solution.

Suppose the initial instance $I$ is transformed by Phase I
to a modified instance $I'$.
We prove the following properties about an optimal solution for $I'$.

\begin{proposition}
An optimal solution for $I'$ has the same number of active slots as an optimal
solution for the original instance $I$.
\label{prop:preprocess}
\end{proposition}
 
\begin{proof}
It is easy to see that any feasible solution for $I'$ is feasible for $I$ since pre-processing 
only created a more constrained instance in the transformation; each job's
window in $I'$ is a subset of its window in the original instance $I$.
We now argue that a solution for $I$ can be transformed to a solution for $I'$
using the same number of slots. 
Suppose a feasible schedule $\sigma$ for $I$ is infeasible for $I'$ due to 
a job $x$. In other words, $\sigma$ schedules job $x$ after its modified deadline in $I'$. 
We can argue  this step by step, by showing that decrementing the
deadline of a single job does not change the optimal solution;
since the modification is done by a sequence of such operations, the optimal solution
is preserved.
Assume that the deadline of $x$ was reduced by one.
In the instance $I'$, we have $B$ jobs with deadline $d_x$, out of which at
most $B-1$ jobs can be scheduled with $x$. Hence there is at least one job
scheduled earlier whose deadline is still $d_x$. Since its release time
cannot be before the release time of $x$, we can exchange these two jobs.
This makes the schedule feasible for $I'$, and this establishes the proposition.
\qed
\end{proof}

\begin{proposition}
Without loss of generality an optimal solution for $I'$ uses a subset of slots that are deadlines.
\label{prop:deadlines}
\end{proposition}

\begin{proof}
Among all optimal solutions select the one that uses the least number of non-deadline
slots.  Among all the slots that are not deadlines, choose the slot $t$ that is 
the right-most (i.e. latest) active slot. Let $X$ be the set of jobs 
which are assigned to $t$.  We now move $X$ as late as possible while maintaining 
feasibility for every job in $X$. We will end up next to a deadline slot $t'$ containing 
a set of blocking jobs $B$; otherwise we will have reduced the number of 
non-deadline slots. From the set of jobs $X \cup B$, select $B$ jobs with the earliest
deadlines and schedule them in $t'$. Since at most $B$ jobs have deadline $t'$,
the remaining jobs all have higher deadlines and the process can be repeated.
\qed
\end{proof}

\begin{theorem}
Algorithm Lazy Activation computes an optimal solution for $I'$ (with the smallest number
of active slots).
\end{theorem}

\begin{proof}
It is enough to prove that the optimal solution w.l.o.g. schedules exactly the same jobs
that we schedule in the first active slot (earliest deadline $d_1$). By removing these jobs 
from the instance $I'$, it is easy to see that our algorithm computes an optimal solution.
The proof for the claim is as follows.
Due to the above propositions, the jobs with deadline $d_1$ w.l.o.g. are all scheduled 
in time slot $d_1$ in the optimal solution; Algorithm Lazy Activation also schedules all the jobs in this time slot. In addition, 
observe that to fill the remaining slots among the available slots, the optimal solution
again w.l.o.g. selects available jobs with earliest deadline, i.e. the filler jobs chosen by 
Algorithm Lazy Activation. Otherwise we can exchange the jobs to achieve this property. 
The proof follows by induction.
\qed
\end{proof}
\noindent
%
%

\subsection{On Infeasible Instances}

In this section, we consider the behavior of
the Lazy Activation algorithm on instances for which
it is impossible to schedule all jobs
within their individual windows of feasibility.
We will show that Lazy Activation maximizes the number
of jobs satisfied.  
In fact, we will see
that it does so in the fewest number of active timeslots.
Denote by $\mathcal{S}_{LA}$ the schedule returned
by Lazy Activation.
In Phase I, it is possible for a job's deadline to be
decremented all the way to its release time; in this
case, we say that the job's window has \textit{collapsed}.
%
%

\begin{proposition}
On infeasible instances, Lazy Activation maximizes
the number of jobs satisfied.
\end{proposition}

\begin{proof}
One can use an argument similar to Proposition 1
to show that even if Phase I collapses the windows
of some jobs, it does not change the maximum
number of jobs that can be scheduled.  For completeness'
sake, we detail it here.  As before, we argue step-by-step that
each decrement changing the instance from $I$ to $I'$
does not affect the maximum throughput.  Suppose
a deadline $d_x$ of job $x$ is reduced by one.  Let $\sigma$
be a feasible schedule on $I$ achieving maximum throughput
on $I$.  We can transform $\sigma$ into $\sigma'$
that is feasible on instance $I'$ and that satisfies
the same number of jobs.  If $\sigma$ does not schedule job
$x$ at $d_x$, then $\sigma$ is already feasible on $I'$.
Suppose $\sigma$ schedules $x$ at $d_x$.  Then $\sigma$
can do at most $B-1$ other jobs at $d_x$.  In $I'$, there
are $B$ jobs with deadline $d_x$ and release time $\ge r_x$.
Thus, there exists a job $j$ which is not scheduled by $\sigma$
at $d_x$ and which has release time at least $r_x$.
Modify $\sigma$ by swapping jobs $j$ and $x$.  (If $j$ was
not scheduled in $\sigma$, then schedule $j$ and not $x$.)  The new
schedule satisfies the same number of jobs and is also
feasible for $I'$.

Thus, the modification of deadlines in Phase I
does not change the maximum number of jobs
which can be scheduled.  In particular, jobs whose
windows collapse in Phase I w.l.o.g. are also
dropped in some throughput-maximizing schedule.
In Phase II, Lazy Activation schedules every job
whose window has not collapse, either at its
deadline or earlier, i.e. as a filler.  Therefore,
Lazy Activation maximizes the number of jobs
satisfied.
\qed
\end{proof}

\begin{proposition}
If Lazy Activation's schedule $\mathcal{S}_{LA}$ satisfies
$n' < n$ jobs in $k$ active slots, then any schedule 
$\mathcal{S}$ satisfying $n'$ jobs does so in at least $k$
active time slots.
\end{proposition}

\begin{proof}
Suppose that Lazy Activation collapses $\kappa = n-n'$ jobs,
denoted $j_{a_1}\ldots j_{a_\kappa}$.
Let $\alpha_t$ be the number of jobs that
had deadline $t$ at the start of the iteration in which Phase I
processed $t$ as the deadline. Deadlines will be decremented
precisely when $\alpha_t > B$. 
For each collapsed job $j_{a_i}$, we will identify an interval
$I_{a_i}$ of excess demand by intuitively
``unrolling'' the iterations of Phase I starting 
from the point of collapse. More formally, set $t'$ to the
latest slot such that every slot $t \in [r_j, t']$ is
such that $\alpha_t > B$. 
Define $I_{a_i}$ to be $[r_j, t')$.
Notice that $t'$ is the original deadline of some job.  
Let $J_{a_i} = \{j : [r_j, d_j) \subseteq I_{a_i} \}$. 
Then the collapsed job $j_{a_i} \in J_{a_i}$
and also $|J_{a_i}| > B \cdot |I_{a_i}|$.  In fact,
$J_{a_i}$ consists exactly of two types of jobs: jobs which
have collapsed and $B \cdot |I_{a_i}|$ jobs which have
not collapsed.  Lazy Activation schedules the latter
job set in $I_{a_i}$ at $B$ jobs per slot.  

Now partition the original instance $(J, T)$ into
two subinstances $(J_1, T_1)$ and $(J_2, T_2)$, where
$J_1 = \bigcup_{i=1}^k J_{a_i}$ and $T_1 = \bigcup_{i=1}^k I_{a_i}$.
Obviously jobs of $J_1$ cannot be scheduled in slots of $T_2$
by definition.  
We claim that since $\mathcal{S}$ maximizes throughput, it
necessarily schedules $J_2$ in $T_2$.
Suppose there exists a job $j \in J_2$ that is
scheduled by $\mathcal{S}$ in some interval $I_{a_i}$.
Then since
$|J_{a_i}| > B \cdot |I_{a_i}|$, there is some job of $J_{a_i}$
that is missed by $\mathcal{S}$. 
Since it is possible to schedule all jobs of $J_2$ only
in slots of $T_2$ ($\mathcal{S}_{LA}$ is such an example),
missing that many jobs of $J_{a_i}$ was unnecessary.
This contradicts the fact that $\mathcal{S}$ maximizes throughput.

Then the active time $A(\mathcal{S})$ of $\mathcal{S}$
($\mathcal{S}_{LA}$, respectively) can be decomposed into
two components: the active time $A_1(\mathcal{S})$ spent
satisfying $J_1$ and the active time spent satisfying $J_2$.
Then,
\begin{align*}
  A(\mathcal{S}) &= A_1(\mathcal{S}) + A_2(\mathcal{S}) \\
	&\ge A_1(\mathcal{S}_{LA}) + A_2(\mathcal{S}_{LA}) \\
	&= A(\mathcal{S}_{LA}) \\
	&= k
\end{align*}
where the inequality follows from the facts that 
(1) Lazy Activation
minimizes active time on feasible instances and
(2) there are exactly $B \cdot |I_{a_i}|$ jobs
contained in each $I_{a_i}$ that have not collapsed 
and Lazy Activation
schedules all of them.  Thus whenever Lazy Activation
is active in $T_2$, it schedules $B$ jobs per slot.
\qed 
\end{proof}

\subsection{Linear Time Implementation}

\def\Ip.{I$'$}

We conclude by showing that the algorithm can be implemented in time $O(n+T)$.
We start by giving the following equivalent version of  Phase I, which we
refer to as Phase \Ip..  Note that in ``assigning'' jobs to
deadlines, both Phases are only preprocessing deadlines.
The actual schedule is not developed until Phase II.

\bigskip

Initially each new deadline value has no jobs assigned to it.
Process the jobs $j$ in order of decreasing release time $r_j$
(jobs with the same release time can be processed in arbitrary order).
Assign $j$ to a new deadline equal to the largest value
less than or equal to $d_j$ that
currently has $<B$ jobs assigned to it.

\bigskip

To prove that Phase \Ip. computes the same deadlines as
Phase I, assume that both algorithms break ties
for release times the same way.  Let $D_j$ and $D_j'$
denote the deadlines to which job $j$ is assigned by
Phase I and Phase \Ip., respectively.  We show that for 
all $j$, $D_j'$ is equal to $D_j$, by induction
on Phase \Ip.'s iterations.
Clearly this holds for the first iteration, since
the first job considered has maximum release time,
so Phase I never decrements its deadline.  Therefore,
both Phases assign it to its original deadline.  
Now consider some job $j$ and
suppose that equality holds for all jobs
processed it iterations previous to that of job $j$.
Then $D_j \le D_j'$:
$D_j$ is an available candidate deadline for Phase \Ip.,
which chooses the latest such one.   Thus, when Phase I assigns
jobs to deadline $D_j'$, it considers assigning
job $j$ (feasible and unassigned) there as well.  
In fact, Phase I assigns to $D_j'$ the jobs that Phase \Ip.
has already assigned there.  There are less than $B$ of them,
each having release time at least $r_j$.  Among the
remaining jobs which are feasible at $D_j'$,
job $j$ has the maximum release time so Phase I
assigns $j$ to $D_j'$.  Therefore, $D_j = D_j'$, as claimed.

We implement Phase \Ip.
using an algorithm for disjoint set merging on the universe
of possible
deadline values $1,\ldots,T$. Each set consists of a deadline $D$
to which $<B$ jobs have been assigned, plus the
maximal set of consecutive values $D+1,\ldots$
that each have $B$ assigned jobs.
It is easy to see
that using the set merging data structure of
\cite{GTSet}
achieves time
$O(n+T)$ for Phase \Ip..

Phase II constructs an earliest-deadline first
schedule on the timeslots that are Phase I deadlines.
In particular, let $D_1<D_2<\ldots <D_\ell$ be the distinct
deadline values assigned to jobs in Phase I.
Modify release times and deadlines to compress the active interval
from $[0,T]$ to $[0,\ell]$, as follows: a job $j$ with
Phase I deadline $D_k$ gets deadline $k$; its 
release time $r_j$
gets changed to the largest integer $i$ such that
$D_i\le r_j$ (take $D_0=0$).
The algorithm computes an earliest-deadline first schedule
on this new instance. It returns the corresponding
decompressed schedule, with
slot $i$ changed back to $D_i$.

Phase II begins by assigning the new release times
in time $O(n+T)$.
Then it constructs an EDF schedule using
disjoint set merging as in Phase \Ip. \cite{Fred},
achieving time  $O(n+T)$ \cite{GTSet}.

%
%
\subsection{Slotted versus Unslotted Model}
Here we show that 
even if the time model is not slotted, the
fact that the release times and deadlines are integral
implies that without loss, the optimal solution schedules
jobs in slots implied by integer time boundaries.

\begin{theorem}
Without loss of generality, the optimal solution
schedules jobs so that they start and end at integral
time points.
\end{theorem}

\begin{proof}
Call a job \textit{non-integral} if it starts (and ends)
at a non-integral time.  Suppose an optimal solution $\mathcal{O}$
has non-integral jobs.  We modify $\mathcal{O}$ to another
optimal schedule $\mathcal{M}$ with fewer such jobs.
Doing this repeatedly eliminates all non-integral jobs.

Let $I$ be the total inactive time of $\mathcal{O}$'s schedule, i.e.,
an optimal schedule maximizes $I$.  In $\mathcal{O}$,
let $[t, t+1)$ be the last interval (for $t$ an integer) 
in which a non-integral job ends.
\newline
\newline
\textit{Case 1: Some processor is busy throughout the interval 
$[t, t+1)$ (e.g., for unit jobs this means some job starts at $t$).}
Construct $\mathcal{M}$ by delaying every non-integral job which 
ends in this slot so that it ends at $t+1$ instead.  $I$ does not decrease,
since every job gets moved into a time period that was already
active in $\mathcal{O}$.  Clearly, the number of non-integral jobs
decreases.
\newline
\newline
\textit{Case 2: No processor is busy during the 
entire interval $[t, t+1)$.}
Let $i$ be the largest integer $< t$ such that an interval
of positive length $[i, i')$ is inactive ($i' > i$).  If no
such interval exists, let $i=0$.

Construct $\mathcal{M}$ by taking every non-integral job that
starts after $i$ and advancing its start time by 
$\epsilon = \min \{ s - \lfloor s \rfloor : s > i 
	\mbox{ the starting time of a non-integral job} \}$.
Obviously $0 < \epsilon < 1$.  Observe that the definition
of $i$ implies than a non-integral jobs advances if and only
if it ends after $i$.

$\mathcal{M}$ is a valid schedule: the choice of $\epsilon$ 
guarantees that no release time is violated.  We claim that
no processor executes more than one job at any point in time, i.e.
no job $j$ advances into the execution of a job $j'$ that does
not advance.  In proof, if $j'$ is integral, then it ends at an
integral time.  The definition of $\epsilon$ shows that $j$ does
not advance past an integral time.  If $j'$ is non-integral,
then it ends at or before $i$.  Again, the definition of 
$\epsilon$ implies that $j$ does not advance past integer time $i$.

To show that $\mathcal{M}$ is optimal, we will show that $I$ does
not change.  Specifically, we show that $I$ increases by $\epsilon$
at the end of the schedule, $I$ does not change in the middle,
and $I$ decreases by $\epsilon$ at the start.

In interval $[t, t+1)$, every job advances by $\epsilon$ (by Case 2),
so $I$ increases by $\epsilon$ in this interval.

Next, consider a maximal interval $[a,b)$ that is inactive in
$\mathcal{O}$ and starts at $a > i$.  We claim that $\mathcal{M}$
is inactive during a corresponding interval $[a', b')$ of the
same or greater length.  The definition of $i$ implies that
$a$ is non-integral, i.e., $a$ is the ending time of a non-integral
job.  This ending time moves to $a' = a - \epsilon$.  If $b$ is
the start time of a non-integral job, that start moves to
$b' = b - \epsilon$.  Clearly, $[a', b')$ is inactive in
$\mathcal{M}$.  The other possibility is that $b$ starts an
integral job.  Setting $b' = b$ gives an inactive interval of
greater length.

If $i=0$, we have shown that $\mathcal{M}$ has more inactive
time than $\mathcal{O}$, a contradiction.  If $i > 0$, the
interval $[i, i')$ shrinks to an interval of inactive time
$[i, i'-\epsilon)$ in $\mathcal{M}$.  This decreases $I$ by
$\epsilon$.  We conclude that $\mathcal{M}$ is an optimal
schedule.  The definition of $\epsilon$ implies that
$\mathcal{M}$ has more integral jobs than $\mathcal{O}$.
\qed
\end{proof}

\section{Unit Jobs with Arbitrary Release Times and Deadlines}
\label{sec:dp}
In this section, we consider a generalization of
the previously discussed problem.
Suppose that the unit length jobs have release times
and deadlines over the reals and that
we want to minimize the number of \textit{batches}
in the schedule, where a batch is a set of jobs which
all start and finish at the same time and where the system
can work on at most one batch at a time.
This objective is slightly more restrictive than active time; 
even so, scheduling unit length jobs with integer
release times and deadlines to minimize active time
is clearly a special case of this, since
in the former problem, the batch property holds without loss.

We describe a simple 
dynamic program which determines in $O(n^8)$ time
the minimum number of batches needed to 
non-preemptively satisfy all jobs, i.e. 
$1|p\mbox{-batch}, B<n, r_i, d_i, p_i=1| K$, 
where $K$ is the number of batches in the schedule.  
The dynamic program (and therefore, the notation) 
is similar to that found in \cite{Baptiste00}, 
on which several DP results in this area of 
scheduling are based.  Suppose throughout that 
jobs are listed in ascending order by deadline, 
i.e. $d_1 \leq d_2 \leq \ldots \leq d_n$.
  
\begin{definition}
  For job $k$ and for interval $[t_\ell, t_r]$ where 
  $r_k \in [t_\ell, t_r]$ and $t_r + 1 \leq d_k$,
  let the set of jobs
  $U_k(t_\ell, t_r) = \{ j \leq k : r_j \in [t_\ell, t_r] \}$ 
  (see Fig.~\ref{DPfig}).  
  Also let $U_0(t_\ell, t_r) = \emptyset$ for all intervals
  $[t_\ell, t_r]$.
\end{definition}

\begin{figure}[t]
\begin{center}

\epsfxsize=5.0in
\mbox{{\epsffile{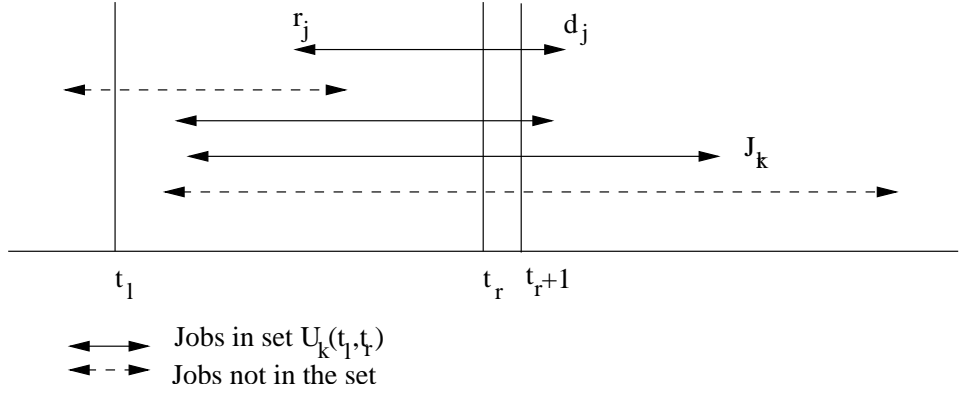}}}

\end{center}
\caption{Jobs of $U_k(t_\ell, t_r)$.}

\label{DPfig}
\end{figure}

Observe that the jobs in $U_k(t_\ell, t_r)$
must be scheduled entirely in $[t_\ell, d_k]$ 
in any feasible schedule.  We restrict ourselves 
without loss to the space of schedules which obey 
the EDF Principle: for any pair of jobs $i$ and $j$ 
where $i < j$, either job $i$ is scheduled before 
or with job $j$, or $r_i$ is after the point at 
which job $j$ is scheduled.  Indeed, if a schedule 
contains a pair of jobs $i$ and $j$ for which this 
does not hold, then one can swap them in the schedule 
without affecting feasibility or the number of batches.

\begin{definition}
  Let $P(k, t_\ell, t_r, \mu_r)$ be the minimum number 
  of batches required to schedule $U_k(t_\ell, t_r)$ such that
  (1) all batches start within $[t_\ell+1, t_r]$ and 
  (2) at most $B-\mu_r$ of these jobs are scheduled in the
      batch starting at $t_r$, and if $\mu_r > 0$, this
      coming at no additional cost.
\end{definition}

If such a schedule does not exist, e.g. $k > 0$ and
$t_\ell+1 > t_r$, then let $P(k, t_\ell, t_r, \mu_r) = \infty$.
Let $P_t(k, t_\ell, t_r, \mu_r)$ be the value
of $P(k, t_\ell, t_r, \mu_r)$ subject to job $k$ 
being scheduled in a batch that starts at time $t$.
Then
\[
	P(k, t_\ell, t_r, \mu_r) = \min_t P_t(k, t_\ell, t_r, \mu_r)
\]
Per Baptiste's observation in \cite{Baptiste}, it is 
enough to iterate over a set of $O(n^2)$ possible starting
times for batches.  For a given $t$, one can compute 
$P_t(k, t_\ell, t_r, \mu_r)$ as follows.  Let 
$L = \{ j \leq k : r_j \in [t_\ell, t] \}$ be the jobs of 
$U_k(t_\ell, t_r)$ which are released before or at $t$.  
Similarly, let $R = \{ j \leq k : r_j \in (t, t_r] \}$ be 
the jobs of $U_k(t_\ell, t_r)$ released after $t$.
The observation above implies that if $i \in L$, then 
$i$ is scheduled before or at $t$.  On the other hand, 
if $i \in R$, then job $i$ must be scheduled after $t$. 
Let $k_L \in \arg \max_{i \in L\backslash \{k\}} d_i$, or 0 if
$L \backslash \{k\} = \emptyset$.
Similarly, let $k_R \in \arg \max_{i \in R} d_i$, or
0 if $R = \emptyset$.

If $t = t_r$, then $L = U_k(t_\ell, t_r)$, $R$ is empty 
and $k_L$ (if positive) is the second-to-latest
deadline in $U_k(t_\ell, t_r)$.  Then,
\begin{equation}
  P_t(k, t_\ell, t_r, \mu_r) = \left\{ 
    \begin{array}{l l}
      P(k_L, t_\ell, t_r, \mu_r+1) & \quad \text{if $0 < \mu_r < B$} \\
      1 + P(k_L, t_\ell, t_r, 1) & \quad \text{if $\mu_r = 0$} \\
      \infty & \quad \text{otherwise} \\
    \end{array} \right.
\end{equation}
\noindent
Consider $t < t_r$.  If $\mu_r > 0$, then for 
$t \in (t_r-1, t_r)$, $P_t(k, t_\ell, t_r, \mu_r) = \infty$, 
since starting batches at $t$ and at $t_r$ will 
violate the constraint that at any given time, 
there is at most one batch running.  Otherwise, 
i.e. if $t < t_r-1$ or $\mu_r = 0$,
\[
P_t(k, t_\ell, t_r, \mu_r) = 1 + 
  P(k_L, t_\ell, t, 1 ) + P(k_R, t, t_r, \mu_r)
\]

Finding the optimal number of batches is equivalent 
to computing $P(n, \min_i r_i - p, d_n - 1, 0)$.
Since there are $O(n)$ jobs and $O(n^2)$ interesting 
times to consider, the total running time is $O(n^7B) = O(n^8)$.
The running time was recently improved to $O(n^3)$
in \cite{Frederick}.

\section{Maximizing Throughput Subject to Batch Size Constraint}
\label{sec:dp2}
\newcommand{\Kscr}{\mathcal{K}}

For infeasible instances, one can apply a 
technique of the same vein to maximize throughput 
subject to an upper bound $\Kscr$ on the number of batches.
Define $P(k, t_\ell, t_r, \mu_r, \kappa)$ 
to be the maximum number of jobs of $U_k(t_\ell, t_r)$ 
which can be scheduled in at most $\kappa$ 
batches starting in $[t_\ell+1, t_r]$, 
where as before, if $\mu_r > 0$, one can schedule up 
to $B-\mu_r$ jobs in the batch starting at 
$t_r$ without incurring cost toward the 
batch budget $\kappa$.

To understand the computation of 
$P(k, t_\ell, t_r, \mu_r, \kappa)$, 
let $t$ be the starting time of the batch which 
satisfies job $k$ (if such a time exists).  
Consider where the rest of $U_k(t_\ell, t_r)$
may be satisfied within $[t_\ell + 1, t_r]$.
Denote by $\alpha$ ($\beta$, respectively) 
the number of them which start in
$[t_\ell+1, t]$ ($(t, t_r]$, respectively).
Observe that $\beta \leq \kappa-\alpha$.  

There are four cases.
Suppose job $k$ is not satisfied in the 
optimal solution.  Then $P(k, t_\ell, t_r, \mu_r, 
\kappa) = P(k', t_\ell, t_r, \mu_r, \kappa)$
where $k' = \arg \max_{i \in U_k(t_\ell, t_r) 
\backslash \{ k \}} d_i$, i.e. the job of next 
latest deadline in $U_k(t_\ell, t_r)$.  (If 
$U_k(t_\ell, t_r) = \emptyset$, then $k' = 0$.)
If job $k$ is satisfied in the schedule, then 
define $L, R, k_L$ and $k_R$ as in Section 3.  
Suppose that job $k$ is satisfied in the 
batch starting at $t=t_r$.  Then
\begin{equation}
P(k, t_\ell, t_r, \mu_r, \kappa) = \left\{ 
  \begin{array}{l l}
    1+P(k_L, t_\ell, t_r, 1, \kappa-1) & \quad \text{if $\mu_r = 0$} \\
    1+P(k_L, t_\ell, t_r, \mu_r+1, \kappa) & \quad \text{if $0 < \mu_r < B$} \\
    -\infty & \quad \text{otherwise} \\
  \end{array} \right.
\end{equation}

If job $k$ is satisfied in a batch which starts 
at $t \leq t_r - 1$, or if $\mu_r = 0$ and $t \leq t_r$, 
then let $\alpha$ be the number of batches used to 
satisfy jobs in $L$; this includes the batch satisfying 
job $k$, so $\alpha > 0$.  Then $\beta = \kappa - \alpha$ 
is the budget on the number of batches used to satisfy $R$.  
Note that for the case where $\mu_r = 0$, if $R$ is 
not empty, then any $\alpha$ corresponding to a 
feasible schedule will be strictly less than $\kappa$, 
i.e. $\beta > 0$.

\begin{equation}
P(k, t_\ell, t_r, \mu_r, \kappa) = \left\{ 
  \begin{array}{l l}
    \max_{0 < \alpha \leq \kappa} (1 + P(k_L, t_\ell, t, 1, \alpha-1) 
      & \quad \text{if $0 \leq \mu_r < B$} \\
    \quad \quad + P(k_R, t, t_r, \mu_r, \kappa - \alpha)) & \quad \\
    P(k, t_\ell, t_r - 1, 0, \kappa )& \quad \text{otherwise} \\
  \end{array} \right.
\end{equation}

Finally, if $\mu_r > 0$ and $t \in (t_r-1, t_r)$, then
$P(k, t_\ell, t_r, \mu_r, \kappa) = -\infty$, since
the machine cannot work on multiple batches at the same
time.  $P(k, t_\ell, t_r, \mu_r, \kappa)$ is the maximum over 
these values.  Since there are  $O(n)$ possible
values for $k$, $O(n^2)$ possible values for both 
$t_r$ and $t_\ell$, and $O(n)$ possible  values
for $\mu_r$, and since $\kappa = O(n)$, there 
are $O(n^7)$ possible values of 
$P(k, t_\ell, t_r, \mu_r, \kappa)$ that need to be computed.  
Computing each one costs $O(n)$, yielding a total 
running time of $O(n^7 \mathcal{K})$.
We note that if the time model is slotted, then there 
are $O(n)$ possible values for both $t_r$ and $t_\ell$, 
and the total running time is instead $O(n^5 \mathcal{K})$.  
This algorithm can be extended to the case in which 
each job $J_k$ has an associated profit $v_k$ and
the goal is to schedule a profit-maximizing set of 
jobs in $\mathcal{K}$ slots.

\section{Disjoint Collection of Windows}
\label{sec:disjoint}
In this section, we discuss the more general problem 
in which each job's feasible regions $T_i$ are arbitrary
(as opposed to being a single interval).
For any fixed $B \geq 3$, we show that the problem
is $NP$-hard and discuss its relationship to other
classic covering problems.  
We also develop an efficient algorithm for 
the two processor case ($B=2$).

\subsection{Connections to Other Problems}
{\bf Capacitated Vertex Cover.}
There is an interesting relationship between the activation
problem of unit-length jobs with multiple feasible
windows and the problem of vertex cover with capacity 
constraints.  In the latter problem, we are given 
a graph $G=(V,E)$ and a capacity function $k(v)$ for 
each vertex.  The goal is to pick a subset $S$ of vertices 
and to determine an assignment of edges to vertices in $S$, 
so that each edge is assigned to an incident vertex in $S$, 
and so that each vertex $v$ has at most $k(v)$ adjacent 
edges assigned to it.  When multiple copies of a vertex 
may be chosen to be part of the cover, a primal-dual 
2-approximation is given in \cite{GHKO}.  In the hard 
capacities version of the problem (VCHC), a bounded 
number of copies of a vertex may be selected.  There 
is a 2-approximation using LP rounding \cite{GHKKS}, 
improving the previous bound of 3 given in \cite{CN}.

In the special case of the activation problem in which
each job has exactly two feasible time slots, there
is an equivalence with VCHC with uniform capacities $k(v) = B$.
Since jobs need not have adjacent feasible time slots, 
one can view time as a set of
slots rather than an ordering of slots.  Then, the
equivalence follows: the slots are the vertices and
the jobs are the edges.  For any valid capacity-respecting
vertex cover of size $C$, there exists a corresponding
activation schedule of cost $C$ (and vice-versa).
One implication of our result 
is that we can solve VCHC optimally when $k(v) = 2$.  
Furthermore, the 2-approximation for VCHC in \cite{GHKKS}
applies to the activation problem when each 
job has two feasible slots.
%
%
There is a similar relationship between VCHC on hypergraphs 
and the activation problem where jobs may have 
more than 2 feasible slots, and for size $g$ hyper edges,
an $O(g)$ approximation has recently been developed
\cite{Barna}. 
\newline
\newline
\noindent
{\bf Capacitated $K$-center.}
Another previously studied problem is the $K$-center 
problem with load capacity constraints \cite{BKP,KS}.  
Given an edge weighted graph satisfying the metric property, 
the goal is to pick $K$ nodes (called centers) and assign 
each vertex to a chosen center that is ``close'' to  it. We 
should not assign more than $L$ nodes to any chosen center 
and want to minimize the value $d_{\max}$ such that each 
node is assigned to a center within distance $d_{\max}$ of it.  
Previous work culminated in a 5-approximation for the 
problem where multiple copies of the same node may be chosen
as a center. If only one copy can be chosen then the bound 
goes up to 6.

There is the following correspondence between this and the 
following instance of the activation problem.  First, guess 
the maximum value $d_{\max}$ and create the (unweighted) 
graph $G^{d_{max}}$ induced by edges having weight at 
most $d_{max}$.  Then create a job $J_i$ and a timeslot 
$t_i$ for each node $i$ in the original graph $G$, and 
let job $J_i$ be feasible in slot $t_j$ if and only if 
node $i$ is within distance $d_{max}$ of node (or center) 
$j$ in $G$, i.e. if edge $(i,j)$ exists in $G^{d_{max}}$.  
(So $t_i$ should be a feasible slot for $J_i$.)  Then, 
finding $K$ centers to which all $n$ vertices can be assigned
in $G^{d_{max}}$ is equivalent to opening $K$ timeslots
and feasibly scheduling all $n$ jobs in them so that no 
slot is assigned more than $B$ jobs, where $B=L$.
Since we can solve the activation problem when $B=2$, we
can optimally solve the $K$-center problem when $L=2$.
\newline
\newline
\noindent
{\bf Capacitated Facility Location.}
Given a set of facilities (with capacities) 
and clients, the goal is to
open a subset of facilities and find
a capacity-respecting
assignment of clients to open facilities that 
minimizes the sum of facility opening costs and 
connection costs.  See \cite{CW,LSS}
for approximation algorithms for this problem. 
Here, we show that when capacities
are uniformly 2, the problem can be solved optimally
in polynomial time.  For clarity's sake,
the details of the connection are described 
at the end of Section~\ref{sub:poly_disjoint}.

\subsection{Proof of $NP$-hardness for B=3}
We prove that the activation problem for jobs of
arbitrary $T_i$ is $NP$-hard when $B=3$ via a reduction from 
3 EXACT COVER, which is known to be $NP$-hard \cite{GJ}.
Given a collection $X$ of $n$ elements and a 
collection of subsets $S_1, \ldots,
S_m$, each containing exactly three elements from $X$. Is there a 
sub-collection of exactly $\frac{n}{3}$ subsets that exactly cover all
of $X$? We can view this problem as a bipartite graph where one side
has the elements $X$ and the other side has a vertex for each
subset $S_i$, and edges denote membership. This problem maps to
the question of finding a dominating set of size exactly $\frac{n}{3}$
where we are only allowed to select from the side containing subsets.  
The relationship with the scheduling problem is now 
obvious: selecting a subset is akin to activating a 
certain time slot where  the set $X$ corresponds to a collection of jobs.
The edges specify in which time slots a job 
can be scheduled. Since the subsets
have size exactly 3, there is an exact cover of size 
$\frac{n}{3}$ if and only if there
is a schedule with exactly $\frac{n}{3}$ active slots.

We observe that there is an 
$O(\log n)$ approximation for this problem (for any $B$) by an easy
reduction to a submodular cover problem. This result follows from
the classical covering results due to Wolsey \cite{Wolsey}.

\def\case #1{\bigskip\noindent{\it Case #1}: }
\def\set #1#2{\{ #1:#2 \}}

\subsection{Polynomial solution for B=2}
\label{sub:poly_disjoint}

This subsection  considers the problem of scheduling jobs
in the fewest number of active slots, when there are two
processors and each job can be scheduled
in a specified subset of time slots. 
We use
the following graph $G$.
The vertex set is $J\cup T$, where $J$ is the set of all jobs
and $T$ the set of all time slots. The edge set is
\[
\set{(j,t)}{\mbox{job $j$ can be scheduled in time slot $t$}} \cup
\set {(t,t)} {t\in T}.
\]
A degree-constrained subgraph (DCS) problem is defined on $G$
by 
the degree constraints $d(j)\le 1$ for every $j\in J$ and
$d(t)\le 2$ for every $t\in T$.
By definition a loop $(t,t)$ contributes two to the degree of $t$.

Any DCS gives a schedule, and any schedule gives a DCS
that contains a loop in every nonactive time slot.  
Also any
DCS $D$ that covers $\iota$ jobs  and contains $\lambda$ loops
has cardinality 
\begin{equation}
\label{DEqn}
|D|=\iota+\lambda.
\end{equation}

\begin{lemma}
\label{MaxCardLemma}
A maximum cardinality DCS $D$ of $G$
minimizes the number of active slots used by any schedule of
 $|V(D)\cap J|$ jobs.
\end{lemma}

\begin{proof}
Let $D$ be
a maximum cardinality DCS.  $D$ 
covers $\iota=|V(D)\cap J|$ jobs. (\ref{DEqn}) shows
no schedule of $\iota$ jobs contains more loops that $D$.
 $D$ contains the loop $(t,t)$ precisely
when $t$ is not active. 
Thus $D$ minimizes the number of active slots for schedules of $\iota$ jobs.
 \qed
\end{proof}

To state the time bounds of this section,
let $n$ be the number of jobs and $m$
the number of given
pairs $(j,t)$ where job $j$ can be scheduled at time
$t$. Observe that $G$ has $O(m)$ vertices and $O(m)$ edges,
since we can assume every time slot $t$ is on some edge $(j,t)$.
In any graph, the majority of edges in a simple non-trivial
path are not loops. So in the algorithms for maximum cardinality 
DCS and maximum matching, an augmenting path has length $O(n)$.
A maximum cardinality DCS on $G$ can be found in time $O(\sqrt n m)$.
The approach is through non-bipartite matching.
We describe two ways to handle the loops. 

The first approach reduces the problem to maximum cardinality
matching in a graph called the $MG$ graph. 
To do this modify $G$, replacing each time slot vertex
$t$ by two vertices $t_1,t_2$. 
Replace each edge $(j,t)$ by two edges
from $j$ to the two replacement vertices for $t$. Finally replace each loop
$(t,t)$ by an edge joining the two corresponding replacement vertices. 

A DCS $D$ corresponds to a matching $M$ in a natural way:
If a time slot $t$ is active in $D$ then $M$
matches corresponding edges of the form  $(j,t_i)$. 
If the loop $(t,t)$ is in $D$ then $M$ matches 
the replacement edge $(t_1,t_2)$.
Thus it is easy to see that a maximum cardinality DCS corresponds to a maximum
cardinality matching of the same size.
The cardinality matching algorithm of Micali and Vazirani 
\cite{MV,GTSet} gives  the desired time bound.
(This approach works for the versions of the problem that we will consider
on unit jobs; it does not work when the jobs have longer length since
a pair of edges $(j,t_1)$, $(j,t_2)$ might get matched.)

A second approach is to use an algorithm for maximum cardinality
DCS on general graphs. If such an algorithm does not handle loops
directly, modify $G$ as follows by replacing each loop $(t,t)$ by
a triangle $(t,t_1,t_2,t)$, where each $t_i$ is a new vertex with
degree constraint $d(t_i)=1$. A DCS in $G$ corresponds to a DCS
in the new graph that contains exactly
$|T|$ more edges, one from each triangle.
The cardinality algorithm of Gabow and Tarjan \cite{GT} gives the desired time bound.

Now let $\iota^*$ be the greatest number of jobs that can be scheduled.
Clearly $\iota^*$ can be computed in time $O(\sqrt n m)$ by finding
a maximum cardinality DCS on the
bipartite graph formed by removing all loops from $G$ \cite{ET}.

\begin{theorem}
\label{I*Theorem}
A schedule for the greatest possible number of jobs ($\iota^*$)
that minimizes
the number of active slots can be found in time $O(\sqrt n m)$.
\end{theorem}

\begin{proof}
The algorithm first
finds a DCS $D_0$ of $\iota^*$ jobs, as described above. It converts $D_0$
 to the desired DCS by using an algorithm to find a maximum cardinality
DCS on $G$, using $D_0$ as the initial DCS. 

The correctness of this approach follows from the fact that
the DCS algorithm
works by augmenting paths. This implies that as the initial DCS
$D_0$ is enlarged to the final DCS, no vertex's degree decreases.
Thus the final solution schedules the same jobs as $D_0$.
Now apply Lemma \ref{MaxCardLemma}.

Using $D_0$ as the initial DCS does not affect the time bound, since
the algorithm begins by finding augmenting paths of length 1,
i.e., an arbitrary set of edges
satisfying the degree constraints \cite{GT,MV}.
\qed
\end{proof}

The proof shows that
the choice of $\iota^*$ jobs is irrelevant --
the minimum number of active slots for a schedule of $\iota^*$ jobs
can be achieved using {\em any} set of $\iota^*$ jobs that can be scheduled.

We also note that matching is a natural tool for our power minimization
problem:
Given a maximum cardinality matching problem, we create a
job for every vertex and if two vertices are adjacent, we create a common time
slot in which they can be scheduled. A maximum cardinality matching 
corresponds to a schedule with minimum active time. 
Thus if $T(m)$ 
is the optimal possible time for our scheduling problem, a maximum cardinality
matching on a graph of $m$ edges can be found in $O(T(m))$ time.
For bipartite graphs
the number of time slots can be reduced from $m$
to $n$. Thus a maximum cardinality matching
on a bipartite graph of $n$ vertices and $m$ edges
can be found in $O( T(n,m) )$ time, for
$T(n,m)$ the optimal time to
find a minimum active time schedule
for $O(n)$ jobs, $O(n)$ time slots,
and $m$ pairs $(j,t)$.

Our algorithm can be extended to other  versions of the
power minimization problem. For example the following corollary models
a situation where 
power is limited. 

\begin{corollary}
For any given integer $\alpha$,
a schedule for the greatest possible
number of jobs using at most $\alpha$ active slots
can be found in time $O(\sqrt n m)$.
\end{corollary}

\begin{proof}
Start by constructing the DCS $D^*$ of
Theorem \ref{I*Theorem} that
schedules
$\iota^*$ jobs and contains say $\lambda^*$ loops.
Let $\iota^*= \tau_1+2\tau_2$, where $D^*$ schedules $\tau_i$ time slots
with $i$ jobs, $i=1,2$. Let $\alpha= \tau_1+\tau_2 -\Delta$.
We consider the following cases for $\Delta$. 

\case {$\Delta\le 0$} $D^*$ has $\le \alpha$ active slots. 

\case {$0<\Delta\le \tau_1$} Choose any $\Delta$ time slots scheduling
just one job and unschedule those jobs. This gives a schedule $D$ with
$\alpha$ active slots. As a DCS $D$ has $\iota^*-\Delta$ jobs
and $\lambda^*+\Delta$ loops.
(\ref{DEqn}) shows
$|D|=|D^*|$. Since fixing the number of active slots fixes $\lambda$,
(\ref{DEqn}) also shows
$D$ schedules the greatest possible number of jobs.

\case {$\tau_1<\Delta$} In $D^*$, unschedule  the jobs that are
scheduled in the $\tau_1$ time slots with one job, or scheduled in
$\Delta-\tau_1$ other time slots (chosen arbitrarily from the $\tau_2$
slots with two jobs).
The result is a schedule $D$ with $\alpha $ active slots, each processing
two jobs. $D$ obviously schedules the greatest possible number of jobs.
\qed\end{proof}

The proof shows that the DCS $D^*$ is easily turned into a table
(of $\tau_1+\tau_2$ entries) that gives the desired schedule for every
value of $\alpha$.
\newline
\newline
\noindent
%
%
Now we highlight the connection to capacitated facility location
in which capacites are two.  Create
a weighted graph $MG_w$,
where jobs correspond to clients
and a facility corresponds to a slot, which is a
pair of nodes $(t_1, t_2)$ in the graph. 
Add edges $(j,t_1)$ and $(j,t_2)$ with 
weights equivalent to the
cost of connecting client $j$ to facility $t$.  Also
put an edge between every pair
of slot nodes $t_1$ and $t_2$, with weight
$-C(t)$ where $C(t)$ is the cost of opening facility $t$.
Then, finding a minimum cost solution to the facility
location problem amounts to determining a minimum cost
matching in $MG_w$ in which each job is scheduled, 
since the former cost is simply the latter plus an
additive term $\sum_t C(t)$.
The matching can be achieved by generalizing our algorithm
to the weighted case in a natural way.

\section{Preemptive Scheduling for Integral Length Jobs}
\label{sec:nonunit}

This section discusses preemptive scheduling (on the 
integer time points only) for $B=2$. Now
each job $j$ has an arbitrary integral length $\ell(j)$,
and a set $T_j$ of time slots in which
one unit of its length can be executed. 
We wish to assign
each job $j$ to exactly  $\ell(j)$ of these time slots.
again minimizing the 
number of active time slots. We state several results for this model
and sketch their proofs.

\begin{theorem}
If a schedule executing 
every job to completion exists,
such a schedule
minimizing the number of active time slots
can be found 
in time $O(\sqrt L m)$ for $L=\sum_j \ell(j)$.
\end{theorem}

\def\PS{\noindent{\bf Proof Sketch:}\ }
\PS  
The algorithm is similar to Theorem \ref{I*Theorem}.
We use the graph $G$, modified
so job $j$ has degree constraint $d(j)\le \ell(j)$.
Any DCS scheduling every job satisfies
(\ref{DEqn}) for $\iota=L$. The rest of the argument
is unchanged. The time bound follows from
\cite{GT}.  
\qed

Now suppose  we cannot schedule all the jobs to completion.
Then it is NP-complete
to even schedule the greatest possible number of jobs (for any fixed number
of processors). We show the proof for $B=1$, the extension for higher $B$ is
trivial.
The reduction is similar to the NP-completeness proof in Subsection 3.1 and is 
as follows: 
We create a time slot for each element in $X$. We create a collection of
$m$ jobs (one corresponding to each set) of length 3 each.  Each job can be
scheduled in the time slots that correspond to the elements in the 
corresponding subset. We can schedule $\frac{n}{3}$ jobs
if and only if there is a solution to the 3 Exact Cover problem.

\def\wg.{{$\widehat G$}}

The previous result assumes that $\ell(j)=3$ for all jobs $j$.
For $B=2$,  the special case where each $\ell(j)$ is 1 or 2 can be modeled by a
graph \wg. which is $G$ augmented by a loop $(j,j)$
for every job $j$ of length 2; also each such $j$ has
degree constraint $d(j)=2$. The following results reduce the problem
to maximum weight DCS.

\begin{theorem}
(a) A schedule for the greatest
possible number of jobs, minimizing the number of active slots,
can be found 
in time  
 $O(\sqrt {n\alpha(m,m)\log^3 m}\ m)$.

(b)  For every  $\iota_1$, 
a schedule of $\iota_1$
unit jobs plus the greatest possible number of length two jobs, 
 minimizing the number of active slots,
can be found in the time bound of (a).

(c) The collection of all schedules of (b)  (i.e., a schedule for every 
possible value of $\iota_1$) can be found in time
$O(n (m +n'\log n'))$, for $n'$ the number of jobs and time slots.
\end{theorem}

\PS
(a) Assign edge weights to \wg.: An edge ($j,t)$ weighs
$m$ if $j$ is unit and $m/2$  if $j$ has length 2. A loop
weighs 1 if it is incident to a time slot and 0 if it is incident to a job.
A maximum weight DCS maximizes the number of jobs, and subject to that,
maximizes the number of idle time slots. The algorithm of
\cite{GT} 
finds a maximum weight DCS in time
$O(\sqrt {n\alpha(m,n)\log^3 n}\ m)$ on a graph
of $n$ vertices, $m$ edges, and integral weights polynomial in $n$.

(b) Add a vertex $x$ adjacent to every unit job. 
Each new edge weighs 0 and $d(x)$, the degree constraint of $x$. is
the number of unit jobs decreased by $\iota_1$.
 
(c) In the graph of (b), set $d(x)=0$ and find
a maximum weight DCS. Then repeatedly increase $d(x)$ by 1 and
update the maximum weight DCS. Each update amounts to finding one 
augmenting path. The algorithm of \cite{GMatch}
finds an augmenting path in time $O(m +n\log n))$
on a graph of $n$ vertices and $m$ edges. 
\qed

Part (c) is motivated by the fact that
unit jobs and length 2 jobs may differ in value.
In the same vein, part (b) obviously generalizes to schedules that maximize
an arbitrary  linear function of the number of unit and length 2 jobs.

%
%
\section{Active Time and Arbitrary Preemption}
\label{PreemptiveSection}

This section considers the case of scheduling a collection of jobs
with arbitrary non-negative lengths $\ell_j$ and
arbitrary sets $T_j$ of feasible intervals, $j\in J$.
There are $B$ processors that operate in a collection of
unit length time slots $s\in S$.
Preemptions are allowed at any time, i.e., job $j$ must be scheduled
in $\ell_j$ units of time
but it can be be started and stopped arbitrarily many times, perhaps
switching processors -- the only constraint is that
it must never execute on more than one processor at any instant of time.
We seek a schedule minimizing the active time.

Note that when we allow preemption
the multi-slot jobs of Section \ref{sec:disjoint} correspond
to arbitrary length jobs. 
Similarly a special case of this model is preemptively scheduling
jobs of arbitrary length with integral release times and deadlines, i.e.,
the generalization of Section \ref{sec:main} which treats
nonpreemptive unit
jobs.

\subsection{Linear program formulation}
\label{LPFormulationSection}
The problem can be formulated as a linear program.
Form a graph of edges $E$ where job $j$ and slot $s$ have an edge
$js$ when $j$ can be scheduled in slot $s$. A variable
$x_{js}$ gives the amount of time job $j$ is scheduled in slot $s$, and
a variable $i_s$ gives the amount of idle (inactive) 
time in slot $s$. The problem is
equivalent to the following LP:

{\parindent=0pt

\bigskip

\begin{equation}
\label{PreEqn}
\global\def\xcns.{(\ref{PreEqn}.\mbox{a})}%
\global\def\icns.{(\ref{PreEqn}.\mbox{c})}%
\global\def\pcns.{(\ref{PreEqn}.\mbox{b})}%
\begin{array}{lllllr}
\mbox{maximize} & \sum_{s\in S}\, i_s\\
\mbox{subject to}& \sum_{js\in E} x_{js} &\ge & \ell_j &j\in J 
&{\hskip0.5in}\xcns.\\
& \sum_{js\in E} x_{js} + Bi_s  &\le  & B &s\in S
&\pcns.\\
&x_{js}+i_s &\le &1&js\in E
&\icns.\\
&i_s,\, x_{js} &\ge &0& s\in S,\,js\in E
\end{array}
\end{equation}

\bigskip

}

To see the formulation is correct first observe that although
\xcns. 
is an inequality we can assume equality holds, since
we need only schedule $\ell_j$ units of job $j$.

Secondly, observe that the
inequalities \pcns.--\icns. are necessary and sufficient
conditions for scheduling $x_{js}$ units of job $j$, $js \in E$, 
in $1-i_s$ units of time (on $B$ processors).
Necessity is clear. For sufficiency,
order the jobs arbitrarily; order the processors arbitrarily too.
Repeatedly schedule the next $1-i_s$ units of jobs on the next  processor,
until all jobs are scheduled. The last processor may not receive
a full $1-i_s$ units of work. Also a job may be split across processors.
However such a job is scheduled last on one
processor and first on the next, so \icns. guarantees
it is not executed simultaneously on two processors. 
\newline
\newline
\noindent
We note that
the following 
LP is equivalent to (\ref{PreEqn}).  Partition $S$ 
into disjoint intervals $[ a,b)$ where $a$ and $b$
are consecutive integers in the sorted set of interval
boundaries over all jobs.
Denote these intervals $I_1, \ldots, I_k$ and
let $E'$ be the set of edges $ij$
where job $j$ is feasible in $I_i$.

{\parindent=0pt

\bigskip

\begin{equation}
\label{PreEqnGenl}
\global\def\xcnsg.{(\ref{PreEqnGenl}.\mbox{a})}%
\global\def\icnsg.{(\ref{PreEqnGenl}.\mbox{c})}%
\global\def\pcnsg.{(\ref{PreEqnGenl}.\mbox{b})}%
\begin{array}{lllllr}
\mbox{maximize} & \sum_{i=1}^k\, y_i\\
\mbox{subject to}& \sum_{ij \in E'} x_{ij} &\ge & \ell_j &j\in J 
	&{\hskip0.5in}\xcnsg.\\
& \sum_{ij \in E'} x_{ij} +B \cdot y_i &\le  & B\cdot |I_i| 
	&i = 1\ldots k &\pcnsg.\\
&x_{ij}+y_i &\le &|I_i|& ij \in E'
	&\icnsg.\\
&y_i,\, x_{ij} &\ge &0& i = 1\ldots k, ij \in E'
\end{array}
\end{equation}

\bigskip

}

The variable $y_i$ denotes the amount of idle 
time in interval $I_i$.  The variable $x_{ij}$ 
denotes the amount of time devoted to job $j$ 
in interval $I_i$.
For jobs of arbitrary length, the size of this LP
is still polynomial in the number of jobs and the 
total number of intervals in the sets $T_j$.

\subsection{The case {\boldmath $B=2$} and 2-matchings}

It is convenient to use  a variant of 
(\ref{PreEqn}) based on  
the matching graph $MG$ of Section \ref{sec:disjoint}.
Each edge of $MG$ has a linear programming variable. Specifically
each edge $js\in E(G)$ gives rise to two variables
$x_e$, $e=js_i$, $i\in \{1,2\}$, that give the amount of time 
job $j$ is scheduled on processor $i$ in time slot $s$.
Also each time slot  $s\in S$, has a variable $x_e$, $e=s_1s_2$
that gives the amount of inactive time in slot $s$.
The problem is
equivalent to the following LP:

{\parindent=0pt

\bigskip

\begin{equation}
\label{Pre2Eqn}
\global\def\jcns.{(\ref{Pre2Eqn}.\mbox{a})}%
\global\def\jscns.{(\ref{Pre2Eqn}.\mbox{c})}%
\global\def\scns.{(\ref{Pre2Eqn}.\mbox{b})}%
\begin{array}{lllllr}
\mbox{maximize} & \sum \set{x_e} {e\in E(MG)}\\
\mbox{subject to}& \sum \set{x_e} {e \hbox{ incident to }j} &= & \ell_j &j\in J 
&{\hskip0.5in}\jcns.\\
& \sum \set{x_e} {e \hbox{ incident to }s_i} &\le  & 1 &s\in S,\ i\in \{1,2\}
&\scns.\\
&\sum \set{x_e} {e\in \{js_1,js_2,s_1s_2\}} &\le &1&js\in E(G)
&\jscns.\\
&x_{e} &\ge &0& e\in E(MG)
\end{array}
\end{equation}

\bigskip

}
 
To see this formulation is correct note that any schedule supplies
feasible $x_e$ values. Conversely any feasible $x_e$'s give a feasible
solution to LP (\ref{PreEqn}) 
(specifically set $x_{js}= x_{js_1}+x_{js_2}$
and $i_s= x_{s_1s_2}$)
and so corresponds to a schedule.
Finally the objective of
(\ref{Pre2Eqn}) equals the total of all job lengths plus the total
inactive time, so like (\ref{PreEqn}) it maximizes the total inactive time.

Now consider an arbitrary graph $G=(V,E)$ and the polyhedron defined by the
following system of
linear inequalities:

{\parindent=0pt

\bigskip

\begin{equation}
\label{2MatchEqn}
\global\def\vcns.{(\ref{2MatchEqn}.\mbox{a})}%
\global\def\tcns.{(\ref{2MatchEqn}.\mbox{b})}%
\begin{array}{lllllr}
&\sum \set{x_e} {e \hbox{ incident to }v} &\le & 1 &v\in V &\vcns.\\
& \sum \set{x_e} {e \hbox{ an edge of  $T$}} &\le  & 1 &
\hbox{$T$ a triangle of $G$}&\tcns.\\
&x_{e} &\ge &0& e\in E
\end{array}
\end{equation}

\bigskip

}
\noindent
Call $\sum \set{x_e} {e \in E}$ the {\em size} of a solution to
(\ref{2MatchEqn}). Say vertex $v$ is {\em covered} if equality holds
in \vcns..
 
Define a  {\em 2-matching} $M$ to be  an assignment of 
weight 0,1 or $1/2$ to each edge
 to each edge of $G$
so that each vertex is incident to edges of total weight
$\le 1$.%
\footnote{This definition of a 2-matching
scales the usual definition by a factor $1/2$, i.e., the weights are 
usually 0,1 or 2.} 
$M$ is  {\em basic} if its edges of weight $1/2$ form a collection
of (vertex-disjoint) odd cycles. The basic 2-matchings are precisely
the vertices of the polyhedron determined by
the constraints (\ref{2MatchEqn}) with \tcns. omitted.
A basic 2-matching is {\em triangle-free} if no triangle
has positive weight on all three of its edges.
Cornu\'ejols and Pulleyblank \cite{CP1}
showed the 
triangle-free 2-matchings are 
precisely
the vertices of the polyhedron determined by
constraints (\ref{2MatchEqn}).

When all job lengths are one,
the  inequalities of (\ref{Pre2Eqn})
are system   (\ref{2MatchEqn})
for graph $MG$, with the further requirement that
every  job vertex  is covered.
So it is easy to see that the result of Cornu\'ejols and Pulleyblank
implies our scheduling problem is solved by
any  triangle-free 2-matching on $MG$
that has maximum size subject to the constraint that it
covers every job vertex.
Also, interestingly,
there is always solution where each job
is scheduled either completely in one time slot or  is split into
two pieces of size $1/2$ (see Fig.~\ref{MGFig}).


\begin{figure}[t]
\begin{center}

\epsfxsize=5in
\mbox{{\epsffile{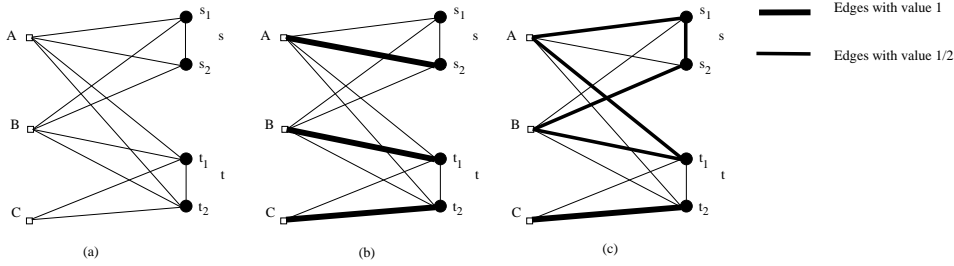}}}

\end{center}
\caption{(a) We illustrate the graph $MG$ with two time slots and three 
jobs with $B=2$. Jobs A and B can be scheduled in time slots s and t, 
while job C can  only be scheduled in time slot t.
(b) We show a solution corresponding to an integral matching and (c) a solution
corresponding to a triangle free 2-matching with $\frac{1}{2}$ unit of idle time in the first
slot.}

\label{MGFig}
\end{figure}

Cornuejols and Pulleyblank give two augmenting 
path algorithms for triangle-free 2-matching:
they find such a matching that is perfect
(i.e., every vertex is fully covered) in time 
$O(nm)$ \cite{CP2}, and such a matching that has
minimum cost in time $O(n^2m)$ \cite{CP1}. (The latter bound 
clearly applies to our scheduling problem, 
since it can model the constraint that every job 
vertex is covered.) Babenko et. al. \cite{COCOON10} 
showed that a maximum cardinality triangle-free 
2-matching can be found in time $O(\sqrt n m)$. 
This is done by reducing the problem to ordinary 
matching, with the help of the Edmonds-Gallai decomposition.

Here we give two easy extensions of 
\cite{COCOON10}: a simple application of
the Mendelsohn-Dulmage Theorem \cite{Lawler} 
shows that a maximum cardinality triangle-free 2-matching
can be found in the same asymptotic time as a 
maximum cardinality matching  (e.g., the algebraic 
algorithm of Mucha and Sankowski \cite{MS04}
can be used). This result extends to
maximum cardinality triangle-free 2-matchings 
that are constrained to cover a given set of
vertices (this is the variant needed for our 
scheduling problem). The development 
is based on standard data structures for blossoms. 
For instance, it is based on a simple definition of
an {\em augmenting cycle} (analog of an 
augmenting path), leading to an {\em augmenting blossom} 
(which models the {\em triangle cluster} of 
Cornuejols and Pulleyblank); we show a maximum 
cardinality matching with the greatest number 
of augmenting blossoms gives a maximum cardinality 
cardinality triangle-free 2-matching 
(Lemma 4(i)).\footnote{Our algorithm was developed 
independently of the 
authoritative algorithm of \cite{COCOON10}.}

Our proof shows that the scheduling problem is solved correctly,
independently of \cite{CP1,CP2,COCOON10}.
We also extend the ideas to get similar results for 
arbitrary job lengths.

\subsection{Triangle-free 2-matchings}

This section gives an efficient algorithm to find a maximum cardinality
triangle-free 2-matching.

\begin{figure}[t]
\begin{center}

\epsfxsize=4.0in
\mbox{{\epsffile{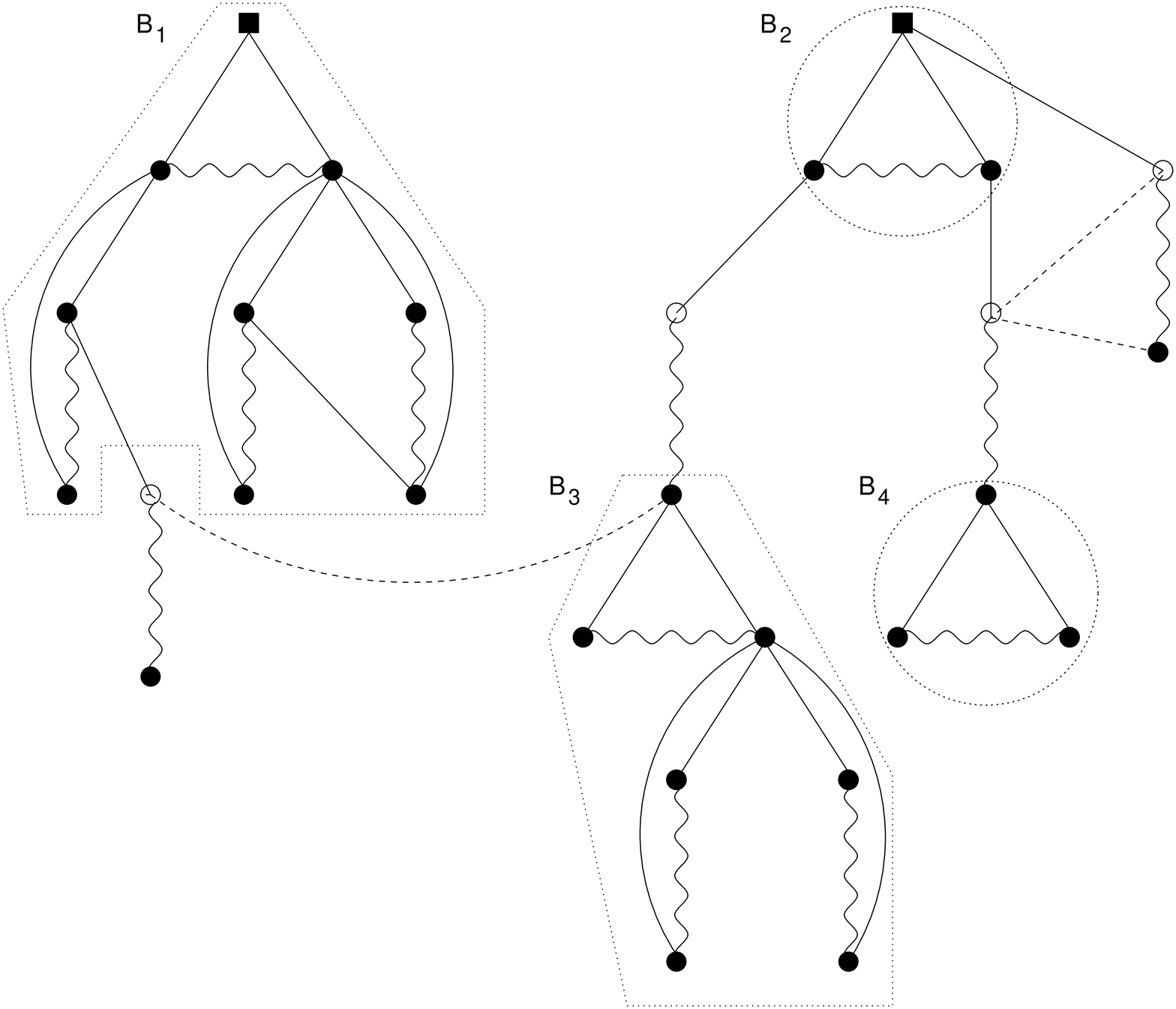}}}

\end{center}
\caption{Matching with four blossoms.}

\label{HungarianFig}
\end{figure}

The approach to 2-matchings is 
through ordinary matchings.
Fig.~\ref{HungarianFig}  illustrates the discussion. Wavy edges are
matched. Square vertices are unmatched.

All paths and cycles that we consider are simple.
A path is  {\em alternating} if its edges are alternately matched
and unmatched; a  cycle is {\em alternating} if one vertex is incident
to two unmatched edges and otherwise the edges are alternately matched and unmatched.

Consider a matching $M$ on an arbitrary graph $G$.  An {\em augmenting
cycle} $A$ is an alternating odd cycle of length $>3$ that contains an
unmatched vertex (i.e., a vertex $\notin V(M)$).  Changing the weight
of every edge of $A$ to $1/2$ gives a 2-matching of size $|M|+1/2$. We
will show that a maximum size triangle-free 2-matching, and a solution
to our scheduling problem, can be constructed using augmenting cycles.

We assume some familiarity with the blossom algorithm for maximum
cardinality matching \cite{Lawler} but we review most notions.

A {\em blossom} is a recursively defined subgraph $B$.  The vertices
of $B$ are partitioned into an odd number of subgraphs $B_i$.  Each
$B_i$ is either a single vertex or a blossom. If we contract each
$B_i$ to a vertex, the remaining edges of $B$ form an alternating odd
cycle spanning every $B_i$. Fig.~\ref{HungarianFig} has four maximal
blossoms $B_1$--$B_4$.

Let $B_0$ be the unique
$B_i$ that is incident to two unmatched edges of $B$.
The {\em base (vertex)} $b$ of $B$
is $B_0$ if $B_0$ is a vertex, else it is the base  of $B_0$. 
An easy induction shows that every vertex of $B-b$ is
on a matched edge of $B$.
If $b$ is on a matched edge, that edge 
is incident to $V(B)$. 
$B$ is a {\em matched (unmatched) blossom} if $b$ is
on some (no) matched edge.

Two well-known properties of blossoms can be shown by simple induction:

\bigskip

{\parindent=20pt
\narrower

{\parindent=0pt

P1: Let $B$ be a blossom with base $b$. For any $v\in V(B)$, $B$ contains
an even length alternating path from $v$ to $b$.

P2: Let $B$ be an unmatched blossom.
The matching on $E(B)$ can be modified
to make any given vertex of $B$ the base.

}}

\bigskip

A {\em
triangle cluster} \cite{CP1,CP2} is 
is a connected subgraph whose biconnected components
are all triangles.  Note that a single vertex
is a triangle cluster.  Also, a subgraph is a triangle
cluster exactly when it is connected, its edges
can be partitioned into triangles,
and any vertex shared by two or more triangles is a cutpoint.
A blossom whose vertices
induce a subgraph that is a triangle cluster is a
{\em t-blossom} ($B_2$--$B_4$ in Fig.~\ref{HungarianFig}); in the opposite case it is 
a {\em non-t-blossom}.  A non-t-blossom
is {\em augmenting} if it is
unmatched ($B_1$ in  Fig.\ref{HungarianFig}). 

\begin{lemma}
\label{AugmentingLemma}
The matching on an augmenting blossom 
$B$ can be modified so $V(B)$ is spanned by
a set of matched edges and an augmenting cycle.
\end{lemma}

\begin{proof}
Consider an arbitrary blossom $A$.
Its recursive definition gives various
subblossoms $B$, each composed of
subgraphs $B_i$, $i=0,\ldots, k$, $k$ even, 
along with edges $u_i v_{i+1}$, $i=0,\ldots, k$
joining $B_i$ to $B_{i+1}$. (Take $k+1$ to be 0.)
Consider one of these subblossoms $B$.

Make $u_0$ the base of $A$ (by P2).
Let $P$ be the even length alternating path in $B$
from $v_1$ to the base $u_0$ (by P1). 
Adding edge $u_0 v_1$ to $P$ completes an alternating cycle $C$.
$A$ is spanned by  $C$ and the
matched edges of $A-C$.
In the first two cases below 
$B$ can be chosen so $|C|\ge 5$,
i.e., $C$ is the desired augmenting cycle.

\case 1 {Some subblossom $B$ has $k>2$.} 
Use $B$ in the above procedure. Note that
$C$ goes through every $B_i$. 
(It
starts in $B_1$, so 
to reach $u_0\in B_0$
it must go through
$B_2,\ldots, B_k, B_0$ in that order.) 
Thus $C$ goes through
$\ge 5$ subblossoms and $|C|\ge 5$.

\case 2 {Some subblossom $B$ has some $u_{i}\ne v_i$.} 
$C$ traverses a path in $B_i$ from $v_i$ to $u_i$.
So it has $\ge 2$ vertices in $B_i$ and
$\ge 1$ vertex in every other  subblossom, i.e., $\ge 4$ vertices.
Since $C$ has odd length, $|C|\ge 5$.

\case 3 {Every subblossom $B$ has $k=2$ as well as $u_i=v_i$ for every $i$.}
Every subblossom consists of three edges forming a triangle.
So the edges of $A$ form a triangle cluster.

Suppose $A$ is an augmenting blossom.
So $A$ does not induce a triangle cluster, i.e.,
some edge $uv\in E(G)-E(A)$ joins two vertices
of $A$. 
Make $u$ the base of $A$. 
Let $P$ be the even length alternating path in $A$
from $v$ to $u$ (by P1). 
Adding edge $u v$ to $P$ completes an alternating cycle $C$.
$A$ is spanned by  $C$ and the
matched edges of $A-C$. $C$ is not a triangle since it contains edge
$uv$ not in the triangle cluster. Thus $|C|\ge 5$.
\qed
\end{proof}

We use some more concepts from ordinary matching \cite{Lawler}.
An {\em alternating tree} has an unmatched root,  every
path from the root is alternating, and every leaf
has even distance from the root. A vertex at even (odd) distance from
the root is {\em outer} ({\em inner}).

\begin{figure}[t]
\begin{center}

\epsfxsize=2.5in
\mbox{{\epsffile{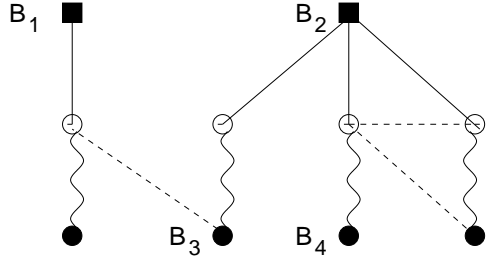}}}

\end{center}
\caption{Hungarian subgraph for the graph in Fig.~\ref{HungarianFig}.}

\label{Hungarian2Fig}
\end{figure}

Fig.~\ref{Hungarian2Fig} illustrates
the following definition (for the matching of 
Fig.~\ref{HungarianFig}):
Let $M$ be a maximum cardinality matching on $G$.
A {\em Hungarian subgraph $H$ 
for $M$} has nodes
that are either vertices of $G$ or contractions of blossoms in $G$.
No vertex of $G$ occurs more than once as a node of $H$ 
or as a member of a blossom. 
Every unmatched vertex of $G$
is either a node of $H$ or the base of an unmatched
blossom of $H$.
$H$ is spanned by a forest of alternating trees.
Each blossom $B$ of $H$ is outer in this forest.
Call a vertex of $G$ outer if it is an outer node
of $H$ or in some blossom of $H$, and inner if it is an inner node.
We sometimes write $V_H$ to denote the set of all inner and outer
vertices of $G$.
The key property is:

\bigskip
{

$(*)$ Any edge $vw$ of $G$ with $v$ an outer vertex 
has $w$ inner or $v$ and $w$ in the same blossom.

}
\bigskip
\noindent
In Fig.~\ref{Hungarian2Fig}  (as well Fig.~\ref{HungarianFig})
inner vertices
are hollow, outer are filled. The three dashed edges
are non-tree edges.

We  use a slight  extension
of this concept: Take $G,M,H$ as before.
Let \B. be the set of augmenting blossoms in $H$.
(We shall see below that \B.
does not depend on choice of Hungarian subgraph $H$.)
Let \Mt. be the triangle-free 2-matching
of size $|M|+|\B.|/2$ formed by 
using Lemma \ref{AugmentingLemma} on 
each blossom of \B. and enlarging $M$ in the obvious way.
A {\em reduced Hungarian subgraph $H$ for \Mt.} 
is a Hungarian subgraph on $G-\bigcup \set {V(B)} {B\in \B.}$
for $M$ restricted to this graph.
We reiterate that a vertex of an augmenting blossom
does not belong to $V_H$. 
The reduced Hungarian subgraph for the matching of Fig.~\ref{HungarianFig}
is Fig.~\ref{Hungarian2Fig}
with $B_1$ deleted,
i.e., the leftmost dashed edge becomes a tree edge.

\begin{lemma}
\label{RedHungarianLemma}
A reduced Hungarian subgraph for \Mt.
satisfies $(*)$.
\end{lemma}

\begin{proof}
For $(*)$ to fail an edge must join an outer vertex $v$ to 
an augmenting blossom $B$. This edge completes an
augmenting path for $M$, from the unmatched vertex of $B$ to
the root of $v$'s alternating tree.
But this contradicts the maximum cardinality of $M$. 
\qed
\end{proof}

This section uses the notion from matching theory
of ``set cover'' (as in odd set covers, reviewed below). 
A {\em set cover} \C. for an arbitrary graph $G$
is a subpartition of the vertices such that every edge of $G$ either has both
ends in the same set of \C.  
or at least one end in a singleton set of \C..
We use several types of these covers, each with its own definition
of ``capacity'' of a set. The intent is that the capacity \C. 
(i.e., the total capacity of all its sets) should upper bound
the size of a matching of some type.

In a {\em triangulated set cover} \C.,
the {\em
capacity} of a set $S\in \C.$ is 
1 for a singleton,
$\lfloor |S|/2 \rfloor$ if $E(S)$ is a
triangle cluster and $|S|/2$ otherwise.
The total weight of any solution to LP (\ref{2MatchEqn}) is at most
the capacity of any such cover \C..  In proof a singleton is incident to
edges of total weight $\le 1$ (its capacity).  Let $S\con V$ be a
non-singleton.  The total weight of all edges with both ends in $S$ is
$\le |S|/2 $.  Finally suppose $S$ induces a triangle
cluster of $t$ triangles.  $S$ has $2t+1$ vertices, and the edges with
both ends in $S$ weigh at most $t=\lfloor |S|/2 \rfloor$.

The 2-matchings we wish to construct
are characterized as follows: 

\begin{lemma}
\label{Max2MatchLemma}

\i If $M$ is a maximum cardinality matching with the greatest
possible number of augmenting blossoms in its Hungarian subgraph,
then $\Mt.$ is a maximum size triangle-free 2-matching.
In fact it is a maximum size solution to (\ref{2MatchEqn}).

\ii Let $J$ be a set of vertices that can be covered by a
triangle-free 2-matching. Let
$M$ be a maximum cardinality matching  with the greatest
possible number of  augmenting blossoms  in its Hungarian subgraph,
subject to the constraint that every vertex of $J$ 
is matched or in an augmenting blossom.
If such an $M$ exists,
 \Mt. is a maximum size triangle-free 2-matching
that covers $J$.
In fact it is a maximum size solution to (\ref{2MatchEqn})
when every vertex of $J$ 
is constrained to have equality in \vcns..
\end{lemma}

\begin{proof}
\i Let $H$ be a reduced Hungarian subgraph for \Mt..
Let $B$ be any blossom in $H$. A maximum cardinality matching 
with $B$ unmatched can be formed 
(by interchanging the matched and unmatched edges in 
the path from $B$ to the root of its alternating tree).
This makes $B$ an augmenting blossom if it is
a non-t-blossom.
So the choice of $M$ makes
$B$ a t-blossom.

Consider the family of sets
\[\C. = \Big\{
V-V_H,\  \{v\},\ V(B) :
\hbox{$v$ an inner vertex,  $B$ a blossom of $H$} 
\Big\}.
\]
\C. is a  set cover (property $(*)$ shows
edges with an outer end are handled correctly).
The size of \Mt. equals the capacity of
\C. considered a triangulated set cover
(since no edge of positive weight leaves $V-V_H$,
any inner vertex is on a weight 1 edge leading to
an outer vertex, and the number of weight 1 edges 
with both ends in a blossom $B$ of $H$ is 
 $\lfloor |V(B)|/2 \rfloor$, which equals the capacity of $V(B)$
as a triangle cluster).
Part \i
follows.

\ii
\Mt. 
covers $J$ since it covers every vertex in an augmenting blossom
of $M$. The rest of the argument is identical to part \xi.
\qed
\end{proof}

We turn to the algorithm to find a maximum cardinality triangle-free
2-matching. First recall two more ideas from matching theory.
An {\em odd set cover} \C. for an arbitrary graph $G$
is 
set cover (as defined above)
where the
capacity of a set $S\in \C.$ is 1 for a singleton, else
$\lfloor |S|/2 \rfloor$  \cite{Lawler,LP}.

It is easy to see that
the cardinality of
any matching is at most the capacity of any odd set cover. 
Furthermore equality can always be achieved:
Take any Hungarian subgraph $H$ for any maximum cardinality matching
$M$. An argument similar to the previous one
shows that $|M|$ equals the capacity of the odd
set cover 

\[ \Big\{
V-V_H,\  \{v\},\ V(B) :
\hbox{$v$ an inner vertex,  $B$ a blossom} 
\Big\}.
\]

Next we recall the Edmonds-Gallai decomposition, which gives the
structure of any maximum cardinality matching. (Parenthetic remarks will
sketch proofs that this decomposition is correct.
Let \C. be an odd set cover whose capacity is the size of a maximum cardinality
matching.) Call a matching
{\em perfect on $S\con V$}
if every vertex of $S$ is matched with another vertex of $S$,
and {\em near-perfect} if this holds for all but one vertex of $S$.
(So $|S|$ is even in the first case and odd in the second.)

Start
with a Hungarian subgraph $H$ for some maximum cardinality matching.
(For our proofs, let \C. be the odd set cover constructed
from $H$ as described above.)
We will describe the structure of an arbitrary maximum cardinality matching
$M$.

\bigskip

{
\narrower

{\parindent=0pt

(a) $M$ contains a perfect matching  of the vertices not in $V_H$.
(This follows since $V-V_H\in \C.$.) 

(b) $M$ contains a near-perfect matching of any blossom $B$ of
$H$. (This follows since $B\in \C.$.)

(c) $M$ contains a maximum cardinality matching of the nodes of $H$.
 (This follows from the maximum cardinality of $M$ and (a)--(b).)
In fact this matching is a maximum cardinality matching of the
bipartite graph $BG$ whose edge sets are the inner nodes, and the
outer nodes of $H$ and whose edges are the edges of $H$ that 
join inner node to outer.
(This follows since \C. shows $M$
does not contain any edge joining two inner vertices.)
In Fig.\ref{Hungarian2Fig} $BG$ does not contain the horizontal dashed edge.

}

}

\bigskip

We remark that the matching of (c) covers every inner vertex (again by \C.).
Note that each edge of $BG$ can be identified with one or more edges of $G$.
Once this edge is chosen  the matching of (c) 
determines the unmatched vertices in the near-perfect matchings of (b).
So starting from any maximum cardinality matching $BM$ of $BG$, we can
construct a corresponding matching of $G$
satisfying (a)--(c), i.e.,
a maximum cardinality matching of $G$.

The last part of Edmonds-Gallai states that the 
sets $V(B)$ for blossoms $B$ of $H$ are an invariant of $G$.
(In proof, property $(*)$ implies
the outer vertices of $H$ are precisely
those vertices that can be reached from an unmatched vertex (of $M$) by an 
even length alternating path. Thus the outer vertices constitute
the set $U$ of
vertices that are unmatched in some maximum cardinality matching.
The sets $V(B)$
are the non-singleton connected components of $U$.)

In the rest of this section,
$O$ denotes the set of outer nodes 
 in the above
bipartite graph $BG$,  
and $N\con O$ denotes its set of non-t-blossoms.

\begin{lemma}
\i Let $BM$ be a maximum cardinality matching of $BG$ that 
matches the greatest possible number of nodes of $O-N$.
Then \Mt. is a maximum size triangle-free 2-matching
if $M$ is a maximum cardinality matching 
of $G$ corresponding to  $BM$ (as described above).

\ii Let $J$ be a set of vertices that can be covered by a
triangle-free 2-matching. Let
$BM$ be a maximum cardinality matching of $BG$ that 
matches every node of $O-N$ that is a subset of $J$,
and as many other nodes of $O-N$ as possible.
Then \Mt. is a maximum size triangle-free 2-matching
that covers $J$
if $M$ is a maximum cardinality matching 
of $G$ corresponding to  $BM$ 
such that no
base vertex 
of an unmatched blossom of $O-N$
belongs to 
$J$.
\end{lemma}

\begin{proof}
\i Let $BM^*$ be a matching on $BG$ that corresponds to
a maximum cardinality matching of $G$ having
the greatest possible number of augmenting blossoms.
The invariance of blossoms implies that the augmenting blossoms
of $BM^*$ are precisely the blossoms of $N$ that are unmatched in $BM^*$.
By Lemma \ref{Max2MatchLemma}
we need only show that $BM$ has as many augmenting blossoms as $BM^*$. 

The size of a matching on $BG$ is the number of matched outer nodes, so
$|O\cap V(BM)| = |O \cap V(BM^*)|$.
By definition
$|(O-N) \cap V(BM)| \ge |(O-N) \cap V(BM^*)|$.
So $|N\cap V(BM)| \le |N\cap V(BM^*)|$. 
This implies the desired inequality
 $|N-V(BM)| \ge |N-V(BM^*)|$.

\ii Suppose
\bigskip

{\narrower

{\parindent=0pt

(a)
the matching $BM$ of part \ii
exists,
 and 

(b) matching $M$ of part \ii exists and  covers every
vertex of $J$.

}}
 
\bigskip
\noindent
Then part \ii follows using
the argument for part \xi.

To establish (a)--(b),
let $J_0$ denote the set of 
nodes of $O-N$ that are  subsets of $J$.

\claim {$BG$ has a matching $BM_0$ that covers 
every node of $J_0$.}

The claim implies (a).
To show it implies (b),
Edmonds-Gallai shows any unmatched node
of $BM$ is outer. 
The claim shows
any blossom of $O-N$ that is unmatched in 
$BM$ contains a vertex $v\notin J$.
$v$ can be made the unique unmatched vertex of $B$ (P2).
So $M$ covers $J-N$, as does \Mt..
\Mt. covers every vertex in a blossom of 
$N$. (b) follows.

To prove the claim
let $M^2$ be a triangle-free 2-matching that covers $J$, viewed
as a collection of matched edges and odd cycles.
We will construct a set of
edges $S\con  M^2$ 
such that
each outer node covered by $M^2$ has degree $\ge 1$ in $S$,
and each inner node has degree $\le 1$.
Clearly $S$ contains the desired matching $BM_0$
($BM_0$ is a subset of $BG$ by property $(*)$).

We construct $S$ by
traversing the connected components of $M^2$, maintaining
the invariant that
any outer node 
covered by a traversed edge of $M^2$ has degree $\ge 1$
in $S$ and any inner node has degree $\le 1$ in $S$.
We traverse a component $C$  of $M^2$ as follows.

Suppose $C$ is a matched edge. Add it to
$S$. Clearly this preserves the invariant.

Suppose $C$ is an odd cycle.
The edges of $C \cap BG$ form a number of connected components.
Traverse each such component, adding alternate edges  to $S$. 
An outer node $x$ in such a component is on $\ge 1$
edge added to $S$,
since $C$ leaves $x$ each time it enters it, and property $(*)$.
An inner node $x$ maintains degree $\le 1$ in $S$, 
since $x$ is on just 2 edges of $C$
(each of which may or may not be in $BG$).
The claim follows.
\qed
\end{proof}

In summary the algorithm to find a maximum cardinality triangle-free
2-matching works as follows.  Find a maximum cardinality matching
$M_0$ of $G$ and its Hungarian subgraph $H$.  Use $H$ to construct
$BG$. Find a maximum cardinality matching $M_1$ of $BG-N$.  Find a
maximum cardinality matching $M_2$ of $BG$ that covers all nodes covered
by $M_1$.  Extend $M_2$ to a maximum cardinality matching $M_3$ of $G$,
using (a)--(c). Convert $M_3$ to the 2-matching $T(M_3)$ by reweighting
the augmenting blossoms.  Return $T(M_3)$.

If the 2-matching is required to cover a set of vertices $J$
two simple changes are needed:
Take $M_1$ as a maximum cardinality matching of $BG-N$
subject to the constraint that
it matches every node of $O-N$ that is a subset of $J$.
Take $M_3$ so it leaves a vertex not in $J$
unmatched in each
 unmatched blossom of $O-N$.

\begin{theorem}
\label{2MatchThm}
A maximum cardinality triangle-free
2-matching can be found in time  $O(\sqrt n m)$.
The same holds if the 2-matching is constrained to 
cover a set of vertices $J$, assuming $J$ can be covered by
some triangle-free
2-matching.
\end{theorem}

\begin{proof}
The maximum cardinality matchings are found in time  $O(\sqrt n m)$
\cite{MV,GTSet}.
A Hungarian subgraph for a given maximum cardinality
matching can be found in time $O(m)$
\cite{GE,GTSet}. Hence the total time for both
of our algorithms is $O(\sqrt n m)$.  
\qed
\end{proof}

\begin{corollary}
A maximum cardinality triangle-free
2-matching can be found in the same asymptotic time
as a maximum cardinality matching.
The same holds if the 2-matching is constrained to 
cover a set of vertices $J$.
\end{corollary}

\begin{proof}
The argument for a 2-matching constrained to
cover $J$ is essentially the same
as the unconstrained case, so we discuss only the latter.
We show that excluding the time for
maximum cardinality matching, our algorithm uses
time $O(m)$.

As just noted $H$ is constructed in $O(m)$ time
in \cite{GE,GTSet}. This  construction 
also supplies the recursive 
decomposition of each blossom $B$ into subblossoms $B_i$,
so it is easy to find the near-perfect
matchings of (b) in Edmonds-Gallai (for $M_3$).
It is also easy to classify each blossom of $H$ as
a triangle cluster or a non-t-blossom.

It only remains to describe
how to find the 
maximum cardinality matching $M_2$ of $BG$ that covers all nodes covered
by $M_1$. The following well-known procedure works for an arbitrary matching
$M_1$ on an arbitrary graph $G$.

Start by finding any maximum cardinality matching $M_2$ of $G$.
The edges of $M_1 \cup M_2$ form a number of connected components,
each of which is a path or cycle with edges alternating between the two matchings.
A node $x\in V(M_1) -V(M_2)$
is the end of a component that is an alternating path $P$.
$P$ has even length
and its other end $y$ is in $V(M_2)-V(M_1)$
(since $M_2$ has maximum cardinality).
So replacing the edges of $P\cap M_2$ by $P\cap M_1$ makes $x$ covered
and keeps $M_2$ maximum cardinality.
($y$ becomes uncovered but this is not a problem.)
Doing this for all such $x$ makes $M_2$ the desired
maximum cardinality matching. Clearly the entire construction uses time $O(n)$.
\qed
\end{proof}

It is well-known that the problems of 
maximum size 2-matching and maximum cardinality bipartite matching
have the same asymptotic time bound.
It is also clear that finding a maximum size triangle-free 2-matching
is at least as hard as  maximum cardinality bipartite matching.

\subsection{Applications to preemptive scheduling 
for {\boldmath $B=2$}}
\label{sec:app}

The algorithm of Theorem \ref{2MatchThm}, along with the 
LP (\ref{Pre2Eqn}), solves the preemptive scheduling problem
for unit jobs. We begin by extending the solution to arbitrary
job lengths $\ell_j$.

We reduce
the general problem 
to the unit length case, using a graph $UG$ with unit jobs
defined as follows.  As before each time slot $s$ is represented by
vertices $s_1, s_2$ and edge $s_1s_2$. A job $j$ of length $\ell_j$,
which may be scheduled in a set of slots $S_j$, is represented by
vertices $u_{js},\ s\in S_j$, each with two edges $u_{js} s_i$,
$i=1,2$. In addition job $j$ has vertices $\bar u_{ji}$,
$i=1,\ldots,|S_j|-\ell_j$.  with a complete bipartite graph joining
its two types of vertices $u_{js}$ and $\bar u_{ji}$.

The triangles of $UG$ have the form
$u_{js}, s_1,s_2$ and are similar to those of $MG$.
Define the sets of vertices
$J_U=\set {u_{js}} {j\in J,\ s\in S_j}$,
$\bar J_U=\set {\bar u_{ji}} {j\in J,\ i\le |S_j|-\ell_j}$.
It is easy to see that a solution to 
LP (\ref{2MatchEqn})
on graph $UG$
that covers every vertex of $J_U \cup \bar J_U$
gives a 
feasible solution to LP (\ref{Pre2Eqn}) 
and vice versa.
So as before  \cite{CP2} 
(or Lemma \ref{Max2MatchLemma}\xii)
shows the 
solution to our problem is given by a 
triangle-free 2-matching on $UG$
that covers $J_U \cup \bar J_U$
and has the greatest cardinality possible.

\begin{theorem}
\label{2PreAlgThm}
Let $J$ be a set of unit jobs  that can be scheduled
on $B=2$ processors. 
A preemptive schedule for $J$ minimizing the
active time can be found in $O(\sqrt n m)$ time. 
The result extends to
jobs of  arbitrary integral lengths $\ell_j$,
where the time is
$O(\sqrt L m)$ for $L$ the sum of all job lengths.
\end{theorem}

\begin{proof}
The algorithm 
for triangle-free 2-matching
(Theorem \ref{2MatchThm})
 can be implemented in
time $O(\sqrt L m)$  on graph $UG$.
To do this we use the vertex substitution technique of \cite{GabDCS}
which works on graphs of $O(m)$ edges rather than $UG$ itself.
\qed
\end{proof}

A standard construction from network flow shows that regarding
feasibility, allowing arbitrary preemption does not help in this sense:
For any number of processors, a set of jobs that 
can be feasibly scheduled with arbitrary preemption can 
also be scheduled limiting preemption to integral
time points, assuming the jobs are either
all unit length, or 
they have arbitrary integral lengths $\ell_j$ and time
is slotted. 
(In detail use a flow graph
where job $j$ is a source of capacity $\ell_j$,
slot $t$ is a sink of capacity $p$,
and edge $jt$ exists when $j$ can be scheduled in slot $t$.
The Integrality Theorem shows we can assume
each $j$ gets scheduled
in exactly
$\ell_j$ time slots.) 

However, even in those situations,
arbitrary preemption can reduce active time.
For example when $B=2$, three unit
jobs that may be scheduled in slots 1 or 2
require 2 units of active time using preemption at integer times
and only $3/2$ units with arbitrary preemption.
We now show this example gives
the greatest disparity between the two preemption models 
for $B=2$.

Recall that we construct an optimal preemptive schedule \P.
by finding a 
a maximum size triangle-free 2-matching
as specified in Lemma \ref{Max2MatchLemma}\ii
(on $MG$ for unit jobs, and $UG$ for arbitrary length jobs)
and converting it to \P. in the obvious way 
via LP (\ref{Pre2Eqn}). Part \iii of the lemma below
gives the main
property for establishing the desired bound.
The following lemma is proved by examining the 
structure of  blossoms in our special graphs, 
e.g., the triangles all have the form $j s_1 s_2$
for $j$ a job vertex and $s_1 s_2$ the edge 
representing a time slot; also vertices $s_1$ and
$s_2$ are isomorphic. 

\begin{lemma}
\label{SNodesLemma}
\i A blossom in $MG$ ($UG$) is a triangle cluster iff it contains 
exactly one vertex of $J$ ($J_U$), respectively.

\ii Let $H$ be a Hungarian subgraph  in $MG$ or $UG$.
Any slot $s$ either has both its vertices
inner, or both its vertices outer (and in the same blossom) 
or neither vertex in $V_H$.

\iii Let the optimal preemptive schedule \P.
be constructed from \Mt., and let
$B$ be an augmenting blossom of $M$.
The time slots with both vertices  in $B$ 
have $\ge 3/2$ units of active time in \P..
\end{lemma}

\begin{proof}
\i A triangle $T$ in $MG$ or $UG$ must contain a slot edge $s_1s_2$.
So the other two edges must be $js_1, js_2$ for vertex $j\in J \cup J_U$.
A triangle $T'$ that shares a vertex but no edge with $T$ must 
contain vertex $j$. So the triangle clusters
are sets of edges $js_1,js_2,s_1s_2$ where $j$ is a fixed vertex of
$J\cup J_U$.
(In $UG$ a cluster has only one triangle.)

Clearly this implies any t-blossom has exactly one vertex of $J\cup J_U$.

Conversely let $B$ be a 
non-t-blossom. Make $B$ an augmenting blossom by unmatching
any edge incident to its base. Then 
Lemma \ref{AugmentingLemma} shows $B$ has an alternating cycle $C$
of $\ge 5 $ edges. $C$ must contain an edge of the form
$s_1s_2$. The two other edges of $C$ incident to $s_1$ and $s_2$
must go to vertices of $J\cup J_U$ that are distinct.

\ii
First observe that  $s_1\in V_H$ implies $s_2\in V_H$: 
If $s_1$ is outer this follows from edge $s_1s_2$. 
If $s_1$ is inner then its parent contains an outer 
vertex  $j\in J\cup J_U$ 
adjacent to $s_1$. $j$ is also adjacent to $s_2$ so $s_2\in V_H$.

We complete the proof by showing that if $s_1$ is outer then so is $s_2$.
If, on the contrary, $s_2$ is inner 
then it is adjacent to an outer vertex in its parent
and an outer vertex in its child. These vertices are obviously in different
blossoms. But $s_1$ is adjacent to both of them, so both are in
the blossom containing $s_1$, contradiction.

\iii First consider
the case of unit jobs and graph $UG$.
Any blossom $B$ in a Hungarian subgraph
contains an even number of slot vertices $s_1,s_2$
(part \xii)
and an odd number of vertices in total.
So 
$B$ contains an odd number of $J$-nodes.
Suppose $B$ is augmenting.
Part \i now implies $B$ has $\ge 3$ $J$-nodes.
\P. schedules the jobs of $B$  in the time slots of $B$.
Clearly they use $\ge 3/2$ units of active time.

Next consider arbitrary length jobs and graph $MG$.
Let $B$ be an augmenting blossom in $UG$.
Lemma \ref{AugmentingLemma} shows $B$ has an alternating cycle $C$
of $\ge 5 $ edges.

Since $|C|$ is odd, $C$ contains an odd number of edges of the form
$s_1s_2$. 
Each such $s_i$ is adjacent in $C$ to a $J_U$-vertex.
The corresponding edge has weight $1/2$ in \Mt., so slot
$s$ has (exactly) $1/2$ unit of active time in \P..
If $C$ has $\ge 3$ such edges $s_1s_2$ then 
the slots of $B$ have $\ge 3/2$ units of active time in \P..
So suppose $C$ contains exactly one $s_1s_2$ edge.

The path in $C$ that avoids 
$s_1s_2$ must contain some slot vertex $t_1$, $t\in S$ (recall
the definition
of $MG$).
Since $t_1t_2\notin C$, $t_1$ is adjacent to two $J_U$-vertices in $C$.
The corresponding edges have weight $1/2$ in \Mt..
So slot $t$ has
1 unit of active time and $s$ has
$1/2$, giving the desired $3/2$ units. 
\qed
\end{proof}

\begin{theorem}
\label{PreThm}
For $B=2$ and a set of jobs $J$, the minimum active
time in a schedule permitting preemption only at
integer times
is $\le 4/3$ times that of one allowing arbitrary preemptions.
\end{theorem}

\remark {The theorem holds for unit jobs, and for jobs
of arbitrary length
$\ell_j$ when time is slotted, i.e.,
a schedule with preemptions only at integer times
is allowed to execute a length $\ell_j$
job in $\ell_j$ distinct time slots, but no more.
If time is continuous the ratio of the theorem 
can approach the trivial bound of $B=2$
when job lengths are arbitrary.
For example consider a length $\ell$ job
with release time 0 and deadline $\ell^2$, and
unit jobs $j=0,\ldots, \ell-1$ with
release time $\ell j$ and deadline $\ell j+1$.
The ratio $(2\ell-1)/\ell$ approaches 2.}

\begin{proof}
An augmenting blossom
increases the size of the 2-matching by $1/2$. 
Recalling the objective function of
LP (\ref{Pre2Eqn}) we see
this increases the number of inactive time units by $1/2$,
i.e.,
it decreases the number of active time units by $1/2$.

Our optimal preemptive schedule \P. is constructed
from \Mt., where the matching $M$ corresponds to
an optimal schedule \N. with preemptions limited
to integer times (recall Lemma \ref{MaxCardLemma}).
Let $M$ have $\pi$ augmenting blossoms
and let \P. have  $\alpha$ active time units.
So \N. has $\alpha+\pi/2$ active time units.

The preceding lemma (part \xiii) shows
$\alpha\ge (3/2)\pi$.
Thus the number of active time units in \N. exceeds
the number in \P. by a factor
\[ (\alpha+\pi/2)/\alpha = 1+ (\pi/2)/\alpha \le 4/3.\]
\qed
\end{proof}


\section{Conclusion}
In this paper, we defined a new problem which involves 
scheduling jobs in batches of size at most $B$.  Each 
job has periods of time within which it must be scheduled 
and the goal is to minimize the number of active time slots 
in the schedule.  No cost is incurred for slots in which 
no jobs are scheduled.  There is a strong connection 
between this problem and other classic covering problems
such as vertex cover with hard capacities and the 
$K$-center problem.  Another general model of energy 
consumption might allow for the energy consumption
of each active slot to depend on the number of jobs 
actually assigned to that time slot.  At least for the 
$B=2$ case this can be handled easily, by adapting the
matching based solution described in Section~\ref{sec:disjoint}.

One could generalize this model further to include 
other objective functions that measure completion 
times, tardyness, etc.  Furthermore, it would be 
interesting to consider the online setting in which the 
entire set of jobs is not known in advance, 
but jobs arrive over time and are known only when 
they are released, or perhaps shortly before 
they are released.  

\end{document}

